\documentclass[journal]{IEEEtran}

\usepackage{amsmath}
\usepackage{tabularx}
\usepackage{graphicx}
\usepackage{subfigure}
\usepackage{amssymb}
\usepackage{color}
\usepackage{amsthm}
\usepackage{multirow}
\usepackage{algorithm,algpseudocode}
\newtheorem{theorem}{Theorem}
\newtheorem{corollary}{Corollary}
\newtheorem{proposition}{Proposition}
\newtheorem{lemma}{Lemma}

\usepackage{url}

\begin{document}

\title{A Game-Theoretic Framework for Resilient and Distributed Generation Control of Renewable Energies in Microgrids }

\author{
          Juntao~Chen,~\IEEEmembership{Student Member,~IEEE,}
       and Quanyan~Zhu,~\IEEEmembership{Member,~IEEE}
\thanks{This paper has been accepted to publish in \textit{IEEE Transactions on Smart Grid}.}
\thanks{This work was supported by the U.S. National Science Foundation under Grants EFMA-1441140, ECCS-1550000 and SES-1541164.}
\thanks{The authors are with the Department of Electrical and Computer Engineering, Tandon School of Engineering, New York University, Brooklyn, NY, 11201 USA. E-mail: \{juntao.chen, quanyan.zhu\}@nyu.edu.}
}

\maketitle

\begin{abstract}
The integration of microgrids that depend on the renewable distributed energy resources with the current power systems is a critical issue in the smart grid. In this paper, we propose a non-cooperative game-theoretic framework to study the strategic behavior of distributed microgrids that generate renewable energies and characterize the power generation solutions by using the Nash equilibrium concept. Our framework not only incorporates economic factors but also takes into account the stability and efficiency of the microgrids, including the power flow constraints and voltage angle regulations. We develop two decentralized update schemes for microgrids and show their convergence to a unique Nash equilibrium. Also, we propose a novel fully distributed PMU-enabled algorithm which only needs the information of voltage angle at the bus. To show the resiliency of the distributed algorithm, we introduce two failure models of the smart grid. Case studies based on the IEEE 14-bus system are used to corroborate the effectiveness and resiliency of the proposed algorithms.

\end{abstract}


\section{Introduction}
Renewable energies, such as solar, wind energy, geothermal and biomass, play an important role in reducing the emission of greenhouse gases and thus are able to mitigate the climate change. Their gradual replacements of the conventional power plants, which generate air polluting by-products including $\mathrm{SO}_2$ and $\mathrm{NO}_x$, are beneficial to reduce the health risks to the human society \cite{ackermann2001distributed}.  According to the energy report Annual Energy Outlook 2015, the total renewable share of all electricity generation increases from 13\% in 2013 to 18\% in 2040 \cite{AEO}. Therefore, it is critical to motivate more participants to generate renewable energies, and thus transform the traditional power grid to a cleaner and more efficient system.

A microgrid is a green system that relies on the renewable distributed resources such as wind turbines, photovoltaics and fuel cells, and it is able to operate independently from the main power grid in an autonomous manner \cite{driesen2008design,hatziargyriou2007microgrids}. Currently, more and more microgrids are integrated with the main power grid for system dependability and resiliency \cite{costa2009assessing,colson2011distributed}. For example, when a generator in the main grid is out of service, microgrids can generate extra power, and inject it into the grid to meet the demand and regulate voltage. Hence, they can play the role of maintaining the stability and reliability of the power system.

 \begin{figure}[!t]
\centering
\includegraphics[width=2.5in]{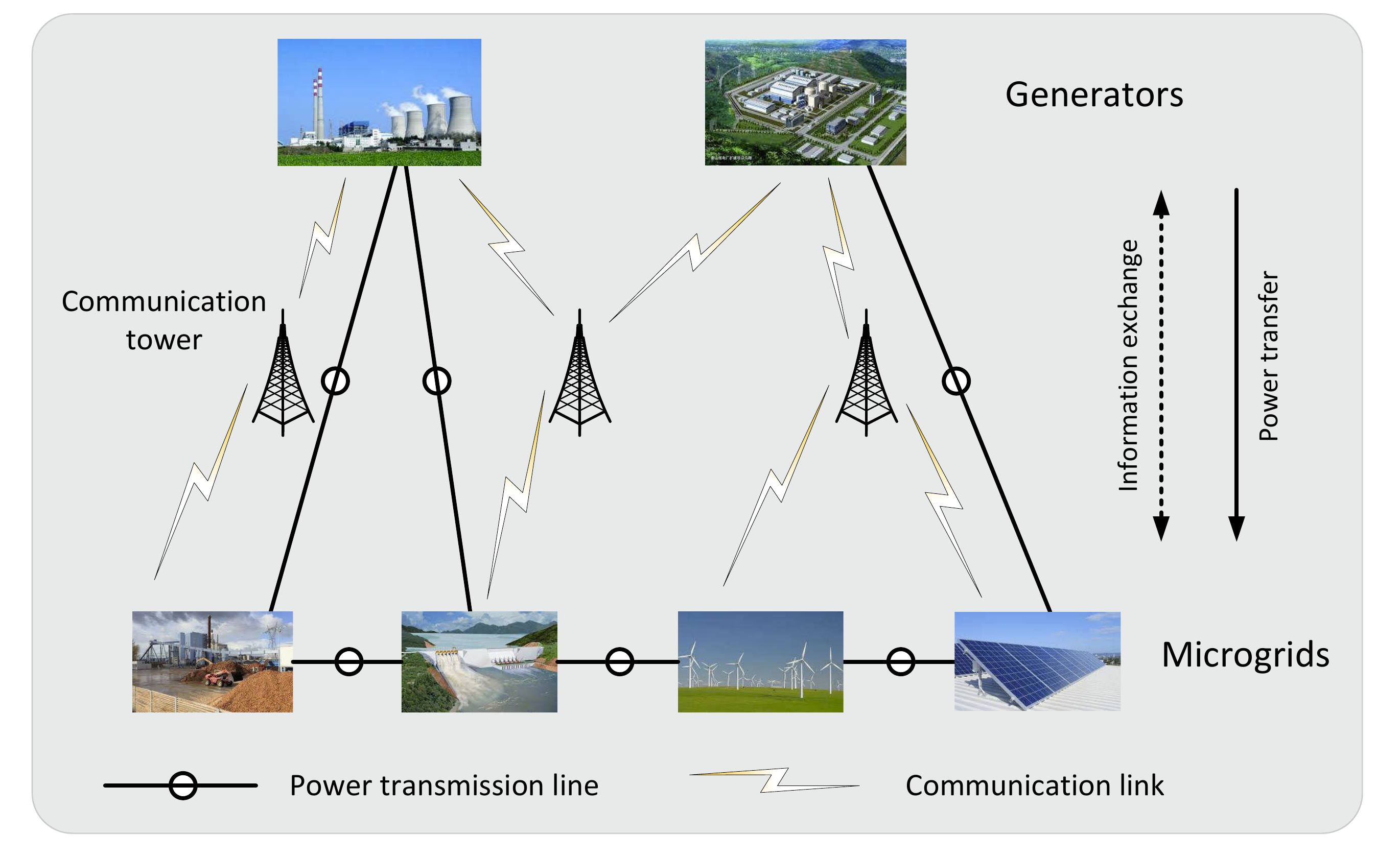}
\caption{Smart grid hierarchy model including the generators, microgrids and communications. Generators in the upper layer determine their amount of power generations and the electricity price and send them to the bottom layer. A microgrid can generate renewable energies and make decisions by responding to the strategies of generators and other microgrids.}\label{twolayer}
\end{figure}

In this paper, we consider the distributed energy management of microgrids when they are integrated with the power grid. With deregulation, future microgrids can enter the electricity market to sell renewable energies at a cheaper price. Reversely, they can also buy electricity from other agents. Each microgrid is making a decision on the amount of power generation to optimize their payoffs. In addition, their strategies are based on the physical and economic factors including the power flow in the system, renewable energy production limit, electricity market price and generation cost \cite{ li2012automated, kasbekar2012pricing, zhu2012game}. The participation of microgrids into the power market, therefore, introduces a competing mechanism between different microgrid entities. Fig. \ref{twolayer} depicts a two-layer system model of the smart grid. Specifically, the upper layer containing conventional generators forms a generator network, and the distributed renewable energy generators in the bottom layer constitutes the microgrid network. The information exchange, such as the electricity market price and the amount of power generation, between two layers are through the communication network lay in the middle. 
On one hand, generators generate power supplied to the loads in the grid, and we assume that their amount is determined ahead of time and thus fixed. On the other hand, microgrids are capable of energy generation to meet their own needs or sell the extra power to the energy market to make a profit.

A natural framework to capture the competition and decentralized decision making of microgrids is game theory \cite{basar1995dynamic,saad2012game}. We use a non-cooperative game-theoretic framework to study the strategic behavior of microgrids in the power grid from a cross-layer perspective. Thus, microgrids are players in this renewable energy game. Different from the centralized control of power systems which is based on the optimal power flow, we adopt Nash equilibrium as our solution concept to characterize the control actions of microgrids. To design an automated energy management system, we develop two update schemes including the iterative update algorithm and random update algorithm which enable the independent decision making of microgrids. However, they need the support from the communication network as depicted in Fig. \ref{twolayer} for the information exchange. To make the control system more intelligent and efficient, we also propose a PMU-enabled distributed algorithm that only depends on the smart device, phasor measurement unit (PMU) \cite{de2010synchronized}, to measure the voltage angles at the buses. By using this scheme, each microgrid does not need to know the specific amount of generation of other players, nor the power supplied by the generators, but can randomly update its strategy only based on the phasor angle at its bus, which preserves high privacy.

The distributed iterative algorithm developed in this work enhances the resilience of the grid by enabling the real-time response to system changes. Power systems tend to fail, such as generator breakdown and trigger of the power transmission lines, due to natural disasters, man-made mistakes and cyber attacks \cite{sridhar2012cyber,glover2011power,kwasinski2012availability}. Therefore, by investigating the two-layer smart grid framework under these faulty scenarios, we can assess the resiliency of the proposed PMU-enabled algorithm.

The contributions of this paper are summarized as follow.
\begin{enumerate}
\item We establish a non-cooperative game-theoretic framework to model the renewable energy generation planning of microgrids in the smart grid, especially by considering the physical and economic constraints.
\item We propose two iterative and random update algorithms for the decision making of microgrids and derive sufficient conditions to ensure their convergence to a unique Nash equilibrium.
\item We develop a resilient and fully distributed PMU-enabled update algorithm for microgrids, and design its implementation control framework. Moreover, two major fault models of the smart grid are introduced to validate the resiliency of the distributed algorithm.
\end{enumerate}

\subsection{Related Work}
Game-theoretic methods have been used for generation planning and control of the distributed energy resources in smart grid \cite{saad2012game,mohsenian2010autonomous,nekouei2015game,fadlullah2014gtes}. Zhu \textit{et al.} \cite{zhu2012game} have proposed a non-cooperative game framework for distributed generation in power systems, and considered economic and AC power flow constraints. Our work extends their model and algorithms to a distributed and resilient fashion. A two-level game-theoretic approach has been used to model the demand response in the cross-layer smart grid framework shown in Fig. \ref{twolayer} which includes both user and utility side \cite{maharjan2013dependable,chai2014demand,zhu2013value}. They also developed distributed algorithms to find the equilibrium solution. However, their main focuses have been on the economic dispatch and generations, and physical constraints, e.g., power flow equations \cite{glover2011power,dall2013distributed}, are not part of their model. Our work is focused on the interactions in the lower layer, and it can be naturally extended to the two-level framework by regarding the generators as active players. 

\subsection{Organization of the Paper}
The rest of the paper is organized as follows. In Section \ref{problem}, we formulate the problem and present a game-theoretic framework to model the generation control in microgrids. We analyze the game in Section \ref{Nash} and propose three update schemes to find the equilibrium solution in Section \ref{schemes}. Implementation framework of the PMU-enabled distributed algorithm is developed in Section \ref{resilience}. Case studies are given in Section \ref{simulation}, and Section \ref{conclusion} concludes the paper. 

\subsection{Notations and Conventions}
Notations and conventions adopted in this paper are summarized as follows: $\mathbf{P}:=[P_1,P_2,...,P_N]'\in\mathbb{R}^N$, $\mathbf{\theta}:=[\theta_1,\theta_2,...,\theta_N]'\in\mathbb{R}^N$; $\mathbf{P}_d:=[P_1,P_2,...,P_{N_d}]'\in\mathbb{R}^{N_d}$, $\mathbf{P}_d^g:=[P_1^g,P_2^g,...,P_{N_d}^g]'\in\mathbb{R}^{N_d}$; $\mathbf{q}:=[q_1,q_2,...,q_{N_d}]'\in\mathbb{R}^{N_d}$; Superscripts $o$ and $*$ indicate the parameter achieving the optimal and equilibrium, respectively; Use of $(n)$ in the superscript indicates the $n$th iteration. In addition, \textit{microgrid} and \textit{player} refer to the same entity, and they are used interchangeably.

\section{Problem Formulation}\label{problem}
In this section, we first review some power flow basics, and then formulate the problem.

\subsection{Power Flow Preliminaries}
In a power grid, let $\mathcal{N}:=\{r,1,2,...,N\}$ be a set of $N+1$ buses, where $r$ denotes the slack bus. Denote $P_i$, $Q_i$, $V_i$ and $\theta_i$ as the amount of active power injection, reactive power injection, voltage magnitude and voltage angle at bus $i$, respectively. Then, the power flow equations of the system with reference to the slack bus $r$ are
\begin{align}
P_i=\sum\limits_{j\in\mathcal{N}} V_iV_j[G_{ij}\cos(\theta_i-\theta_j)+B_{ij}\sin(\theta_i-\theta_j)],\\
Q_i=\sum\limits_{j\in\mathcal{N}} V_iV_j[G_{ij}\sin(\theta_i-\theta_j)-B_{ij}\cos(\theta_i-\theta_j)],
\end{align}
for $i,j=1,2,...,N$, where $G_{ij}$ and $B_{ij}$ are the real and imaginary part of the element $(i,j)$ in the admittance matrix $\mathbf{Y}\in\mathbb{C}^{N\times N}$ of the power grid. Note that $V_r$ and $\theta_r$ are both known, and specifically, $\theta_r=0$.

Let $P_i^g$ and $P_i^l$ be the power generation and power load at bus $i$, respectively. Then, the active power injection at bus $i$ satisfies
\begin{equation}\label{injection}
P_i=P_i^g-P_i^l,\ \forall i\in\mathcal{N}.
\end{equation}
Moreover, by considering the balance of the grid, we have
$\sum_{i\in\mathcal{N}} P_i^g=\sum_{i\in\mathcal{N}} P_i^l.$

DC approximation is usually used for fast calculations of the power flow \cite{glover2011power}. Assume that the reactance is much smaller than the resistance on transmission lines; the voltage angles $\theta_i,\ i\in \mathcal{N}$, are small, and the voltages $V_i,\ i\in\mathcal{N}$, are equal to 1. Then, $Q_i=0,\ i\in\mathcal{N}$, and $\sin(\theta_i-\theta_j)\approx \theta_i-\theta_j$, $\cos(\theta_i-\theta_j)\approx 1$. Therefore, power flow equations can be represented by a set of linear equations as
$
P_i=\sum_{j\neq i}B_{ij}(\theta_i-\theta_j),\ \forall i,j\in\mathcal{N},
$
which can be written in a matrix form as
\begin{equation}\label{power_dc}
\mathrm{\textbf{P}}=-\mathrm{\textbf{B}}\boldsymbol{\theta}.
\end{equation}

\textit{\textbf{Remark 1}}: Matrix \textbf{B} includes the imaginary components of \textbf{Y} \textit{except} the slack bus's row and column. Since $-\textbf{B}$ is a symmetric reduced Laplacian matrix, then, \textbf{B} is invertible by the Kirchhoff's matrix-tree theorem \cite{chaiken1982combinatorial}.

\subsection{Game-Theoretic Framework}
Consider the smart grid is composed of load buses and generator buses. We denote $\mathcal{N}_d:=\{1,2,...,N_d\}\subseteq \mathcal{N}$ as the set of $N_d$ buses that can generate renewable energies such as the wind and solar power. Therefore, buses in the set $\mathcal{N}_d$ are able to generate power for self-efficiency. In addition, denote $\mathcal{N}_g:=\{N_d+1,N_d+2,...,N\}=\mathcal{N}\setminus \mathcal{N}_d$ as a set of $N_g$ buses that are either PQ-loads or generators, where set $\mathcal{N}$ excludes the slack bus $r$ for notation clarity. For load bus $i\in\mathcal{N}$, it has specified load $P_i^l$, and for generator bus $j\in\mathcal{N}$, it has predetermined generation $P_j^g$.

Furthermore, for buses in the set of $\mathcal{N}_d$, they are connected with the microgrids that are able to generate renewable energies. Their loads to be serviced are specified ahead of time, and their generations of renewable energy need to be determined. Note that each microgrid has a maximum generation, i.e., 
$0\leq P^g_i \leq P^g_{i,\max},\ i\in\mathcal{N}_d.$ $P^g_{i,\max}$ depends on the available stored power of the microgrid, and it is dynamically changing due to the intermittent nature of renewable generations. For convenience, the framework is studied for a given $P^g_{i,\max}$ in this paper, and it can be generalized to dynamic $P^g_{i,\max}$ if the state-of-charge of storage devices is considered. In addition, for buses in $\mathcal{N}_g$, we set the power generation of PQ-load bus to 0, and the power load of generator bus to 0 without loss of generality,.

Before formulating the game-theoretic framework for this power generation game, we make several assumptions as follows.
First, the topology of the whole power system is known to all microgrids. This is justifiable since the parameters of power transmission lines, such as resistance and reactance, are often known. Second, the constraints of each microgrid are common information, and each microgrid is aware of the physical constraints when making decisions \cite{levron2013optimal,wan2010optimal}. This indicates that every microgrid should take the power flow constraints \eqref{power_dc} into account. Third, power generations are given for all generators, and microgrids take actions by responding to the generator network. This is reasonable since microgrids can regulate itself more quickly than the generators, and they can be viewed as followers who respond to the generators \cite{zhu2012game,maharjan2013dependable,chai2014demand}. We also assume that PMU can be employed to measure the voltage angle at the bus for all microgrids which is already a mature technology in the smart grid \cite{de2010synchronized}.

Let  $G:=\big\{\mathcal{N}_d,{\{ \mathcal{P}^{\mathcal{G}}_i,\Theta_i\} }_{i\in \mathcal{N}_d},{\{ U_i \}}_{i\in \mathcal{N}_d}, \mathcal{P} \big\}$ be a strategic game with a set $\mathcal{N}_d$ of $N_d$ players. ${\{ \mathcal{P}^{\mathcal{G}}_i,\Theta_i\} }$ is the action set of player $i$, where
$\mathcal{P}^{\mathcal{G}}_i:=\{ P_i^g\in\mathbb{R}_+\ |\ 0\leq P_i^g \leq P^g_{i,\max} \}$, 
and $\mathcal{P}$ is the feasible set of active power injection defined by constraint \eqref{power_dc}. For convenience, denote $\{\mathcal{P}^{\mathcal{G}}_{-i},\Theta_{-i} \}$ by the Cartesian product of all players' action sets except $i$'th one. In addition, denote $\mathcal{P}^{\mathcal{G}}$ by the feasible set of power generation of all buses in the grid, which can be obtained by using $\mathcal{P}$ through \eqref{injection}. Then, the feasible renewable generation set of game $G$ can be defined by $\mathcal{P}^{\mathcal{G}}_F:=\big( \otimes_{i\in \mathcal{N}_d} \mathcal{P}^{\mathcal{G}}_i \big)\cap\mathcal{P}^{\mathcal{G}}$. In addition, based on the set $\mathcal{P}^{\mathcal{G}}_F$, we can obtain the feasible voltage angle profile of all players $\Theta_F$ through \eqref{power_dc}. Thus, the feasible action set of all players $\{ P_i^g,\theta_i\}_{i\in\mathcal{N}_d}$ can be defined as $\mathcal{F}:=\mathcal{P}^{\mathcal{G}}_F \times \Theta_F$.  Note that game $G$ is coupled through the power flow equations.

The cost function $U_i:\mathcal{P}_i^{\mathcal{G}}\times [0,\pi]\rightarrow \mathbb{R}$ for player $i$ is given by
\begin{equation}\label{utility}
U_i(P_i^g,\theta_i)=\psi_i P_i^g+\zeta(P_i^l-P_i^g)+\frac{1}{2}\eta_i^2\theta_i^2,\ \ i\in \mathcal{N}_d,
\end{equation}
where $\psi_i$ is the unit cost of generated power for player $i$, $\zeta$ is the unit price of renewable energy for sale defined by the power market, and $\eta_i$ is a weighting parameter that indicates the importance of regulations of voltage angle at bus $i$. Note that in \eqref{utility}, player $i$ prefers a small voltage angle $\eta_i$ which satisfies a condition of using DC approximated power flow.

Then, the optimization problem for player $i$ is
\begin{align*}
\mathrm{OP}_i: \quad\min\limits_{P_i^g,\theta_i}\quad &U_i(P_i^g,\theta_i)\\
\mathrm{s.t.}\quad \mathrm{\textbf{P}}&=-\mathrm{\textbf{B}}\boldsymbol{\theta},\\
0\leq& P_i^g \leq P_{i,\max}^g,\ i\in\mathcal{N}_d.
\end{align*}

The renewable generation game $G$ is a constrained game. A Nash equilibrium (NE) solution pair $(\mathbf{P}_d^{g*},\boldsymbol{\theta}_d^*)$, where $\mathbf{P}_d^{g*} = [P_i^{g*}]_{i\in\mathcal{N}_d},\ \boldsymbol{\theta}_d^* = [ \theta_i^* ]_{i\in\mathcal{N}_d}$, is a point where no player can benefit from deviating from it through changing his strategy. The formal definition of NE of game $G$ is as follows.

\textit{\textbf{Definition 1} (Nash Equilibrium of Game $G$ with Coupled Constraints)}: The solution pair $(\mathbf{P}_d^{g*},\boldsymbol{\theta}_d^*)$ constitutes a \textit{Nash equilibrium} point for the microgrids in the non-cooperative game ${G}$ if, for $\forall i\in \mathcal{N}_d$,
$${U}_i(P_i^{g*},\theta_i^*)\leq {U}_i(P_i^g,\theta_i),\ \forall (P_i^g,\theta_i)\in \Phi_i(P_{-i}^{g*},\theta_{-i}^*),$$ 
where $\Phi_i(P_{-i}^{g*},\theta_{-i}^*)$ is a projected constraint set defined by
$\Phi_i(P_{-i}^{g*},\theta_{-i}^*):=\{ (P_i^g,\theta_i):(P_i^g,\theta_i;P_{-i}^{g*},\theta_{-i}^*) \in\mathcal{F}  \}$.
 
\subsection{Team Problem}
For comparison, we formulate a team problem (TP) that captures the global optimality of the smart grid as 
\begin{align*}
\mathrm{TP}: \quad\min\limits_{P_i^g,\theta_i}\quad &\sum_{i\in \mathcal{N}_d} \alpha_i U_i(P_i^g,\theta_i)\\
\mathrm{s.t.}\quad &\mathrm{\textbf{P}}=-\mathrm{\textbf{B}}\boldsymbol{\theta},\\
&0\leq P_i^g \leq P_{i,\max}^g,\ i\in\mathcal{N}_d,
\end{align*}
where $\alpha_i\in (0,1),i\in \mathcal{N}_d$, are the weights on microgrids such that $\sum_{i=1}^{N_d}\alpha_i=1$. Note that the weighting constant $\alpha_i$ indicates the importance of microgrid $i$ among all players, and it is chosen or assigned by a system coordinator in the smart grid.

Denote the optimal solution to the TP as $\mathbf{P}_d^{go}=[P_i^{go}]_{i\in\mathcal{N}_d},\ \boldsymbol{\theta}_d^o = [\theta_i^o]_{i\in\mathcal{N}_d}$. Then, the loss of efficiency (LOE) due to the decentralized decision making is defined as
\begin{align*}
\mathrm{LOE}:= \frac{\sum_{i \in \mathcal{N}_d} \alpha_i U_i(P_i^{g*},\theta_i^*)}{\sum_{i \in \mathcal{N}_d} \alpha_i U_i(P_i^{go},\theta_i^o)}.
\end{align*}

\textit{\textbf{Remark 2}}: The value of LOE satisfies $0<\mathrm{LOE}\leq 1$. 

\section{Analysis of the Game}\label{Nash}
Before finding the solution to the formulated optimization problem in Section \ref{problem}, we reformulate and analyze the renewable generation game in this section.
  
Since \textbf{B} is nonsingular, then, \eqref{power_dc} can be rewritten as 
\begin{equation}\label{theta}
\boldsymbol\theta=\textbf{SP},
\end{equation}
where \textbf{S}:=$[s_{ij}]_{i,j\in \mathcal{N}}=-\textbf{B}^{-1}$. The property of elements in matrix \textbf{S} is summarized in lemma 1.

\begin{lemma}\label{lemma_Smatrix}
 \textbf{S} is a symmetric matrix, and ${s_{ij}\geq 0},\ \forall\ i,j\in\mathcal{N}$, and especially $s_{ii}>0,\ \forall\ i\in\mathcal{N}$.
 \end{lemma}
 
\begin{proof}
See Appendix \ref{appLemma}.
\end{proof}

Power flow equations \eqref{power_dc} and \eqref{theta} are equivalent, then, problem $\mathrm{OP}_i$ for player $i$ can be simplified as follows.

\begin{proposition}
The optimization problem $\mathrm{OP}_i$ for player $i$ is equivalent to the following problem:
\begin{align*}
\mathrm{OP}_i': \quad\min\limits_{P_i^g,\theta_i}\quad &U_i(P_i^g,\theta_i)\\
\mathrm{s.t.}\quad &\theta_i=\sum_{j\in \mathcal{N}} s_{ij} P_j,\\
&0\leq P_i^g \leq P_{i,\max}^g,\ i\in\mathcal{N}_d.
\end{align*}
\end{proposition}
\begin{proof}
The decision variables in the objective function \eqref{utility} for player $i$ are $P_i^g$ and $\theta_i$. In addition, equation $\theta_i=\sum_{j\in \mathcal{N}} s_{ij} P_j$ includes all the information related to the decision variables in the original power flow constraint $\mathrm{\textbf{P}}=-\mathrm{\textbf{B}}\boldsymbol{\theta}$, and thus $\mathrm{OP}_i$ and $\mathrm{OP}_i'$ are equivalent for $\forall i\in\mathcal{N}_d$.
\end{proof}

Next, by plugging constraint $\theta_i=\sum_{j\in \mathcal{N}} s_{ij} P_j$ into the objective function, we obtain
\begin{equation}\label{utility2}
\widetilde{U}_i(P_i^g,P_{-i}^g)=\psi_i P_i^g+\zeta(P_i^l-P_i^g)+\frac{1}{2}\eta_i^2 (\sum_{j\in \mathcal{N}} s_{ij} P_j)^2,
\end{equation}
for $i\in \mathcal{N}_d$, where $\widetilde{U}_i:\mathcal{P}_i^{\mathcal{G}} \times \mathcal{P}_{-i}^{\mathcal{G}} \rightarrow \mathbb{R}$, and it is strictly convex over $P_i^g$. Note that after simplification, the objective function \eqref{utility2} is only related to the amount of power generation. Then, we can define a new game $\widetilde{G}:=\big\{\mathcal{N}_d,{\{ \mathcal{P}^{\widetilde{\mathcal{G}}}_i\} }_{i\in \mathcal{N}_d},{\{ \widetilde{U}_i \}}_{i\in \mathcal{N}_d}, \mathcal{P} \big\}$ which is equivalent to game $G$. Note that the action set $\mathcal{P}^{\widetilde{\mathcal{G}}}_i$ of player $i$, $i\in \mathcal{N}_d$, and the feasible set $\mathcal{P}^{\widetilde{\mathcal{G}}}_F$ of game $\widetilde{G}$ is the same as the one in game $G$, respectively.
In addition, power flow equations are not in the constraints of $\widetilde G$ anymore. Then, game $\widetilde G$ can be categorized to a generalized Nash equilibrium problem \cite{rosen1965existence,facchinei2010generalized}. Instead of directly applying the general results from \cite{rosen1965existence,facchinei2010generalized}, we analyze the renewable generation game by considering its unique characteristics to obtain more insights. The definition of Nash equilibrium of game $\widetilde G$ is as follows.

\textit{\textbf{Definition 2} (Nash Equilibrium of Game $\widetilde G$)}: The set of renewable generation profile $\mathbf{P}_d^{g*}$ constitutes a \textit{Nash equilibrium} point for the microgrids in the non-cooperative game $\widetilde{G}$ if 
$$\widetilde{U}_i(P_i^{g*},P_{-i}^{g*})\leq \widetilde{U}_i(P_i^g,P_{-i}^{g*}),\ \forall P_i^g\in \mathcal{P}^{\widetilde{\mathcal{G}}}_i,\ i\in \mathcal{N}_d.$$ 

Note that for a given $\mathbf{P}_d^{g}$, there exists a unique corresponding voltage angle profile $\boldsymbol{\theta}_d$. Thus, there is a one-to-one correspondence between $\mathbf{P}_d^{g}$ and $\boldsymbol{\theta}_d$. Therefore, the Nash equilibrium solution to game $G$, $(\mathbf{P}_d^{g*},\boldsymbol{\theta}_d^*)$, is strategically equivalent to the NE solution to game $\widetilde G$, $\mathbf{P}_d^{g*}$.
Then, by using the first-order optimality condition to \eqref{utility2}, we obtain
$
\psi_i-\zeta+\eta_i^2(\sum_{j\in \mathcal{N}} s_{ij} P_j)s_{ii}=0,\ \ i\in \mathcal{N}_d.
$
By defining $g_i:=s_{ii}P_i$, $\bar{g}_{-i}:=\sum_{j\neq i\in \mathcal{N}}s_{ij}P_j$ and since $s_{ii}\neq 0,\ i\in \mathcal{N}_d$, we have
$g_i=\frac{\zeta-\psi_i}{\eta_i^2 s_{ii}} - \bar{g}_{-i},\ \ i\in \mathcal{N}_d,$
which is equivalent to
\begin{equation}\label{userpower}
P_i=\frac{1}{s_{ii}}(\frac{\zeta-\psi_i}{\eta_i^2 s_{ii}} - \bar{g}_{-i}),\ \ i\in \mathcal{N}_d.
\end{equation}

For convenience, define a player specific parameter $\gamma_i$ for $i$th player as 
$\gamma_i:=\frac{\zeta-\psi_i}{\eta_i^2 s_{ii}},$
and denote $P_i^{\max}=P_{i,\max}^g-P_i^l$. To find the equilibrium solution, we express the set of fixed point equations \eqref{userpower} in a matrix form
\begin{equation*}
\begin{bmatrix}
1&\frac{s_{12}}{s_{11}}&\frac{s_{13}}{s_{11}}& \dotsm &\frac{s_{1N_d}}{s_{11}}\\[6pt]
\frac{s_{21}}{s_{22}}&1&\frac{s_{23}}{s_{22}}& \dotsm &\frac{s_{2N_d}}{s_{22}}\\[6pt]
\vdots & \vdots & \vdots & \ddots & \vdots\\[6pt]
\frac{s_{N_d 1}}{s_{N_d N_d}}& \frac{s_{N_d 2}}{s_{N_d N_d}} & \dotsm & \frac{s_{N_d N_{d}-1}}{s_{N_d N_d}} & 1
\end{bmatrix}
\begin{bmatrix}
P_1^*\\[6pt]
P_2^*\\[6pt]
\vdots\\[6pt]
P_{N_d}^*
\end{bmatrix}=
\begin{bmatrix}
q_1\\[6pt]
q_2\\[6pt]
\vdots\\[6pt]
q_{N_d}
\end{bmatrix}
\end{equation*}
\begin{flalign}\label{matrixform}
&\Leftrightarrow \quad \textbf{H}\textbf{P}_d^{*}=\textbf{q},&
\end{flalign}
where $\textbf{H}:=[ \frac{s_{ij}}{s_{ii}} ]_{i,j\in\mathcal{N}_d} ,\ \textbf{q} := [q_i]_{i\in \mathcal{N}_d} = [ \frac{\gamma_i}{s_{ii}} -\sum_{j\in \mathcal{N}_g} s_{ij}P_j]_{i\in \mathcal{N}_d} $, and $\textbf{P}_d^{*}:= [P_i^*]_{i\in \mathcal{N}_d}=[P_i^{g*}-P_i^l]_{i\in \mathcal{N}_d}$.

On one hand, if the solution to \eqref{matrixform} is an inner point that satisfies $0\leq P_i^g \leq P_{i,\max}^g,\ \forall i\in \mathcal{N}_d$, then it is a feasible optimal solution. On the other hand, if player $i$'s payoff attains its minimum at a generation level out of the feasible interval, then, the optimal solution will be achieved at the boundary point. Specifically, $P_i$ has the following form:
\begin{align}\label{netpower}
P_i=\begin{cases}
\begin{array}{ll}
-P_i^l, &\mathrm{if}\ \gamma_i \leq \bar{g}_{-i}-s_{ii}P_i^l,\\
P_i^{\max}, &\mathrm{if}\ \gamma_i \geq \bar{g}_{-i}+s_{ii}P_i^{\max},\\
\frac{1}{s_{ii}}(\gamma_i - \bar{g}_{-i}), &\mathrm{otherwise}.
\end{array}
\end{cases}
\end{align}
The above results are summarized in Theorem \ref{uniqueNE}.

\begin{theorem}\label{uniqueNE}
The renewable energy generation game $\widetilde G$ admits a unique Nash equilibrium, and the net power injection of player $i$ to the grid is given by \eqref{netpower}.
\end{theorem}
\begin{proof}
See Appendix \ref{app1}.
\end{proof}

Next, we analyze the connection between the team optimal solution and Nash equilibrium solution as follows.

\begin{corollary}
 The optimal solution to the team problem is identical with the game solution, and thus $\mathrm{LOE}=1$.
\end{corollary}

\begin{proof}
Since $(\mathbf{P}_d^{go},\boldsymbol{\theta}_d^o)$ is an optimal solution to the team problem, then, for $ \forall t\in \mathcal{N}_d$,
$$ \sum_{i\in\mathcal{N}_d} \alpha_i{U}_i(P_i^{go},\theta_i^o)\leq   \sum_{j\neq t \in\mathcal{N}_d} \alpha_j {U}_j(P_j^{go},\theta_j^o)+ \alpha_t U_t(P_t^g,\theta_t),$$
where $P_t^g\in\mathcal{P}^{\mathcal{G}}_t,\ \theta_t\in \Theta_t$. Therefore,
$\alpha_t {U}_t(P_t^{go},\theta_t^o)\leq \alpha_t {U}_t(P_t^{g},\theta_t)$, and $(P_t^{go},\theta_t^o)$ is a solution to $\mathrm{OP}_i$, $\forall t\in\mathcal{N}_d$.
 Note that $(\mathbf{P}_d^{go},\boldsymbol{\theta}_d^o)$ satisfies the power flow constraints, and thus it is in the projected constraint set in Definition 1.  Then, it directly follows that $(\mathbf{P}_d^{go},\boldsymbol{\theta}_d^o)$ constitutes a NE of game $G$. Since game $G$ is equivalent to game $\widetilde G$, $\mathbf{P}_d^{go}$ is a NE of game $\widetilde G$. Based on Theorem \ref{uniqueNE}, we obtain $\mathbf{P}_d^{go}=\mathbf{P}_d^{g*}$. Therefore, $(\mathbf{P}_d^{go},\boldsymbol{\theta}_d^o)=(\mathbf{P}_d^{g*},\boldsymbol{\theta}_d^*)$ and $\mathrm{LOE}=1$.
\end{proof}

\textit{\textbf{Remark 3}}: From Corollary 1, we know that the solution obtained via the decentralized decision making is as efficient as that obtained by the centralized control of microgrids.

\section{Update Schemes for the Generation Game}\label{schemes}
We have analyzed the existence and uniqueness of the Nash equilibrium of the renewable energy generation game in Section \ref{Nash}. In this section, we first present iterative update algorithm (IUA) and random update algorithm (RUA) which are based on the gradient-descent method \cite{nocedal2006numerical}, and then design a PMU-enabled distributed algorithm to compute the NE solution of the game. The iterative algorithms enable the adaptivity and resilience of the grid in response to disruptive events. We will show their convergence to NE solutions and study their rates of convergence.

\subsection{Iterative Update Algorithm}
The iterative update algorithm is a scheme that each player updates their amount of power generation simultaneously at time step $n$ which is given by
\begin{align}
P_i^{(n+1)}&=\Psi_i(\gamma_i,\bar g_{-i}^{(n)}) \notag\\
&=\min\Big( P_i^{\max},\ \max\big[-P_i^l,\frac{1}{s_{ii}}(\gamma_i - \bar{g}_{-i}^{(n)})\big]\Big) \notag\\
&=\min\Big( P_i^{\max},\ \max \big[ -P_i^l,\notag\\
&\qquad\qquad \frac{1}{s_{ii}}(\gamma_i - \sum\limits_{j\in \mathcal{N}_g}s_{ij}P_j 
-\sum\limits_{j\neq i\in \mathcal{N}_d}s_{ij}P_j^{(n)} )\big]\Big).\label{iterativeupdate}
\end{align}

The IUA converges to the unique Nash equilibrium
$
P_i^{*}=\min\Big( P_i^{\max},\ \max \big[ -P_i^l,\frac{1}{s_{ii}}(\gamma_i - \sum_{j\in \mathcal{N}_g}s_{ij}P_j
-\sum_{j\neq i\in \mathcal{N}_d}s_{ij}P_j^{*}) \big]\Big),
$
which is equivalent to
\begin{align*}
P_i^{g*}=\min\Big( P_{i,\max}^{g},\ \max \big[ 0,\frac{1}{s_{ii}}(&\gamma_i - \sum\limits_{j\in \mathcal{N}_g}s_{ij}P_j\\
&-\sum\limits_{j\neq i\in \mathcal{N}_d}s_{ij}P_j^{*} )+ P_i^l \big]\Big)
\end{align*}
from any feasible initial point of $P_i^g,\ \forall i\in \mathcal{N}_d$. A sufficient condition that ensures the global stability and convergence of IUA is summarized in Theorem \ref{stability}. Note that the IUA algorithm implicitly handles the local capacity constraints and the coupled network constraints in $\mathrm{OP}_i$. Hence, each iteration yields a feasible solution of $\mathrm{OP}_i$.

\begin{theorem}\label{stability}
The iterative update algorithm is stable and converges to the unique Nash equilibrium if the following condition holds
\begin{align}
\max_{i,j\neq i\in \mathcal{N}_d} \frac{s_{ij}}{s_{ii}}(N_d-1)<1.\label{iterativetheom}
\end{align}
\end{theorem}
\begin{proof}
See Appendix \ref{app2}.
\end{proof}

The IUA requires that microgrids update their actions synchronously. In cases where synchronization is not possible, we develop a generalized mechanism in the next subsection to capture the asynchronous and random updates.

\subsection{Random Update Algorithm}
When no synchronization mechanism exists between players, one more practical update scheme is random update algorithm. More specifically, the players update their generations of the renewable energy in the discrete time intervals with a predefined probability $0<\tau_i<1,\ i\in\mathcal{N}_d$. The random update algorithm is 
\begin{equation}\label{randomupdate}
P_i^{(n+1)}= \begin{cases}
\begin{array}{ll}
{\Psi_i(\gamma_i,\bar g_{-i}^{(n)})}, &\mathrm{with\ probability}\ \tau_i,\\
P_i^{(n)}, &\mathrm{with\ probability}\ 1-\tau_i,
\end{array}
\end{cases}
\end{equation}
where $\Psi_i$ is defined in \eqref{iterativeupdate}.

For the stability and convergence of the random update algorithm, we have the following theorem.
\begin{theorem}\label{thm3}
The random update algorithm is globally stable and converges to the unique Nash equilibrium almost surely if the following condition holds
\begin{align}
\bar{\tau}\cdot \max_{i,j\neq i\in \mathcal{N}_d} \frac{s_{ij}}{s_{ii}}(N_d-1)<\underline{\tau},\label{randomthem}
\end{align}
where $\bar{\tau}$ and $\underline{\tau}$ are the upper bound and lower bound of the probability $\tau_i$, respectively, for $\forall i\in \mathcal{N}_d$.
\end{theorem}
\begin{proof}
See Appendix \ref{app3}.
\end{proof}

\textit{\textbf{Remark 4}}: When all players have the same probability of update, i.e., $\tau_i=\tau,\ \forall i\in\mathcal{N}_d$, then, \eqref{randomthem} is reduced to $\max_{i,j\neq i\in \mathcal{N}_d} \frac{s_{ij}}{s_{ii}}(N_d-1)<1$, which is the same as \eqref{iterativetheom}.

In the IUA and RUA, the players compute their optimal renewable energy generations based on the market electricity price provided by the generator network, and the specific generated power of all other players and generators in the grid. In order to know these information, one possible solution is that all microgrids and generators send their amount of generation to a data center, and then the center broadcasts the received information. Therefore, IUA and RUA actually highly depend on the communication networks which are costly and not confidential. These drawbacks motivate us to design a distributed and convenient update scheme in the ensuing subsection.

\subsection{PMU-enabled Distributed Algorithm}\label{tech_algorithm}
 In our proposed PMU-enabled distributed algorithm (PDA), players do not need to share their private information, e.g., active power injection, and do not need to know the power supplied by the generators. Moreover, PDA does not need the synchronization mechanism as that in IUA since they update their strategies at any time interval. Hence, its update fashion is similar to RUA but requires much less information. Note that in the smart grid, power flow equation $\boldsymbol\theta=\textbf{SP}$ leads to
\begin{equation}\label{measure}
\sum\limits_{j\in \mathcal{N}_g}s_{ij}P_j+\sum\limits_{j\neq i\in \mathcal{N}_d}s_{ij}P_j=\theta_i-s_{ii}P_i,\ \forall i\in\mathcal{N}_d.
\end{equation}
Therefore, based on \eqref{iterativeupdate} and \eqref{measure}, one way for player $i$ to update his action at time $n$ is to know the current voltage angle $\theta_i^{(n)}$ at his bus which can be measured by PMU, and his net power injection update scheme is given by
\begin{align}
P_i^{(n+1)}=\min\Big( P_i^{\max},\ &\max \big[ -P_i^l,\notag\\
&\frac{1}{s_{ii}}(\gamma_i -\theta_i^{(n)} +s_{ii}P_i^{(n)} )\big]\Big).\label{pmu_update}
\end{align}

The only required knowledge for each player in this scheme are the electricity price and smart grid topology which are both known. Moreover, this distributed algorithm is stable and converges to the unique Nash equilibrium almost surely. The proof is similar to that of Theorem \ref{thm3} only by changing $\bar g_{-i}^{(n)}$ to $\theta_i^{(n)} - s_{ii}P_i^{(n)}$ and thus omitted here. For clarity, PDA is summarized in Algorithm \ref{algorithm3}.

\textit{\textbf{Remark 5}}: When a microgrid updates its generated renewable energy, it will consequently change the voltage angle profile $\{\theta_i\}_{i\in\mathcal{N}}$ in the smart grid. Therefore, $\theta_i^{(n)}$ in \eqref{pmu_update} will be changed due to other players' updates which makes the distributed algorithm feasible.

\begin{algorithm}[!t]
\caption{PMU-enabled distributed algorithm}\label{algorithm3}
\begin{algorithmic}[1]
\State Initialize $P_i^{(0)}\gets 0,\ P_i^l,\ P_i^{\max},\ \tau_i,\ \forall i\in\mathcal{N}_d$, $P_j,\ j\in\mathcal{N}_g$, tolerance $\delta>0$, $P^{(1)}\gets 2\delta$, $n\gets 1$
\While {$\lVert \textbf{P}_d^{g(n)} - \textbf{P}_d^{g(n-1)} \rVert_\infty > \delta$}
\While {$i\leq N_d$}
\State Generate a random number $k$ between $[0,1]$ with uniform distribution
\If {$k\leq\tau_i$}
\State PMU measures voltage angle $\theta_i^{(n)}$ at bus $i$
\State Obtain $P_i^{(n+1)}$ through \eqref{pmu_update}
\Else
\State $P_i^{(n+1)}\gets P_i^{(n)}$
\EndIf
\State $i\gets i+1$
\EndWhile
\State $n\gets n+1$
\EndWhile
\State $P_i^{g*}=P_i^{(n)}+P_i^l,\ i\in\mathcal{N}_d$
\State \textbf{return} $P_i^{g*},\ i\in\mathcal{N}_d$
\end{algorithmic}
\end{algorithm}

\subsection{Convergence Rate Analysis}\label{rateanalysis}
One method to measure the efficiency of the update algorithm is its convergence rate which can be quantified by the contraction mapping constant. For the IUA and RUA, the contraction constants are equal to
\begin{align*}
c_1&=\max_{i,j\neq i\in \mathcal{N}_d} \frac{s_{ij}}{s_{ii}}(N_d-1),\\
c_2&=\bar\tau\cdot \max_{i,j\neq i\in \mathcal{N}_d} \frac{s_{ij}}{s_{ii}}(N_d-1) + (1-\underline\tau).
\end{align*}

The smaller $c_1$ and $c_2$ are, the faster the algorithms converge to the Nash equilibrium, and more efficient for the microgrids to respond to the  power generation changes from the generator network. For both IUA and RUA, smaller values of $\max_{i,j\neq i\in \mathcal{N}_d} \frac{s_{ij}}{s_{ii}}$ and $N_d$ are desirable. The two parameters are related to the smart grid topology and the number of microgrids, respectively. A smaller $N_d$ can be interpreted as a smaller number of players, which makes it easier to reach a consensus. In addition, the values of $\bar \tau$ and $\underline \tau$ also have an impact on the convergence speed of RUA. Specifically, a smaller $\bar \tau$ and a larger $\underline \tau$ result in faster convergence. 

\textit{\textbf{Remark 6}}: Since $\bar \tau \geq \underline \tau$, then, for a given $\bar \tau$ or $\underline \tau$, the best outcome is achieved at $\bar \tau = \underline \tau$, which means that every player has the same update probability.

\section{Resiliency and Implementation of the Distributed Algorithm}\label{resilience}
The implementation of the fully distributed algorithm proposed in Section \ref{tech_algorithm} merely requires the information of market electricity price and smart grid topology. To study the resiliency of the algorithm, we focus on two fault models of the power system network. In addition, we develop a control framework to implement the algorithm in the smart grid.

\subsection{Fault Models of Smart Grid}
Without loss of generality, we discuss two major types of faults in the smart grid: (i) generator breakdown \cite{glover2011power}, and (ii) open-circuit of the transmission line \cite{kezunovic2011smart}. Specifically, the breakdown of generator $i$ can be captured by 
$P_i^g=0,\ i\in\mathcal{N}_g.$
Note that microgrid can be turned off during the operation, and this case can be captured by $P_i^g=0,\ i\in\mathcal{N}_d.$
The open-circle fault of the transmission line between buses $i$ and $j$ could lead to
$B_{ij}=0,\ P_{ij}=P_{ji}=0,\ i\neq j\in\mathcal{N},$
where $P_{ij}$ denotes the power flow from bus $i$ to bus $j$. For the balance and stability of the smart grid, generator or microgrid outage will increase the power generation of some other entities in the system. While the trip of a transmission line results in the re-dispatch of power in the smart grid.

Other failure and attacks models can also be studied using this framework including data injection attacks \cite{rawat2015detection,kim2011strategic}, unavailability of PMU data \cite{reinhard2012data}, and jamming attacks \cite{mo2012cyber,bou2013communication}. The fault models chosen here represent the major physical consequences of the cyber attacks.

\begin{figure}[!t]
\centering
\includegraphics[width=0.85\columnwidth]{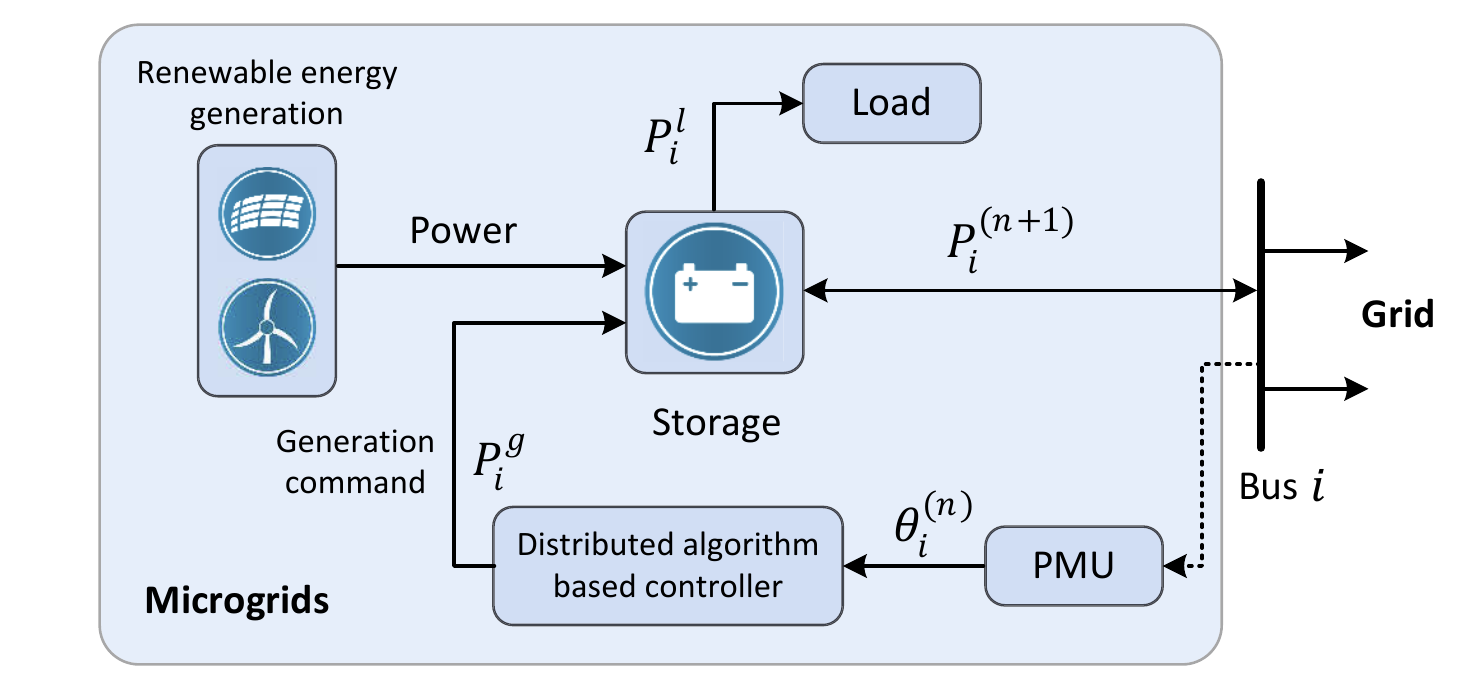}
\caption{The framework to implement the PMU-enabled distributed algorithm. PMU measures the voltage angle at the bus, and the controller generates a command regarding the amount of microgrid renewable energy injection from the local storage to the grid based on the received voltage angle.}\label{implementation}
\end{figure}

\subsection{Implementation Framework}
The proposed framework used to implement the PMU-enabled distributed algorithm is shown in Fig. \ref{implementation}. On one hand, renewable energy generators in the microgrids produce power and store them in the storage devices. On the other hand, the PMU measures the voltage angle $\theta_i^{(n)}$ at bus $i$ at time interval $n$, and sends it to the controller. Then, the controller generates a command based on the PMU-enabled distributed algorithm that informs the storage device to inject $P_i^{(n+1)}$ amount of renewable energy to bus $i$. Note that negative $P_i^{(n+1)}$ indicates that player $i$ buys power from the grid. 

In the framework, the power storage can be seen as unchanged during the implementation stage if the decision updates of microgrids are fast comparing with the physical storage dynamics. Under the condition that the microgrids make strategies over a period of time,  then the storage dynamics need to be considered, and the algorithm is generalized to the dynamic setting.

\textit{\textbf{Remark 7}}: The PMU-enabled algorithm is applicable at buses where PMU is installed. For nodes without PMU sensors, IUA and RUA can be used for the update. Then, the algorithm can be generalized to a hybrid one.

\section{Case Studies}\label{simulation}
In this section, we validate our proposed algorithms via case studies based on the IEEE 14-bus system. Fig. \ref{14bus} shows the power system model. Buses 3, 6 and 14 are three players in the power generation game, and they are connected to microgrids that generate wind, solar, geothermal renewable energies, respectively. Generation bus 2 is selected as the slack bus which serves as a basis of the power system and also absorbs the power uncertainties in the grid \cite{glover2011power,sauer1998power}. Power transmission line parameters of the 14-bus system can be found in \cite{rich}. Without loss of generality, the weighting parameters $\eta_i=\eta$ are the same for $\forall i\in \mathcal{N}_d$, and the generation capacity of microgrid $i$ is equal to $P_{i,\max}^g=100\mathrm{MW},\ \forall i \in\mathcal{N}_d$. Other parameters for the case studies are summarized in Table \ref{t1}. For convenience, the number in the subscript of each parameter denotes the bus indexing in Fig. \ref{14bus}. 
\begin{table}
\centering
\renewcommand\arraystretch{1.3}
\caption{Numerical Value of Parameters in Simulation\label{t1}}
\begin{tabular}{|c|c|c|c|} \hline
Parameter&Value&Parameter&Value\\ \hline
$\psi_3$ & $120\$/\rm{MWh}$ & $P_3^l$ & $120\rm{MW}$\\ \hline
$\psi_6$ & $100\$/\rm{MWh}$ & $P_4^l$ & $100\rm{MW}$\\ \hline
$\psi_{14}$ & $80\$/\rm{MWh}$ & $P_6^l$ & $105\rm{MW}$ \\ \hline
$\zeta$& $140\$/\rm{MWh}$& $P_8^l$& $110\rm{MW}$\\ \hline
$P_1^g$& $280\rm{MW}$ & $P_{10}^l$& $90\rm{MW}$\\ \hline
$P_{11}^g$& $160\rm{MW}$  & $P_{12}^l$& $85\rm{MW}$\\ \hline
$\eta$ & $3\times 10^4$& $P_{14}^l$ & $70\rm{MW}$ \\
\hline
\end{tabular}
\end{table}

\subsection{Effectiveness of Algorithms}
 \begin{figure}[t]
\centering
\includegraphics[width=0.65\columnwidth]{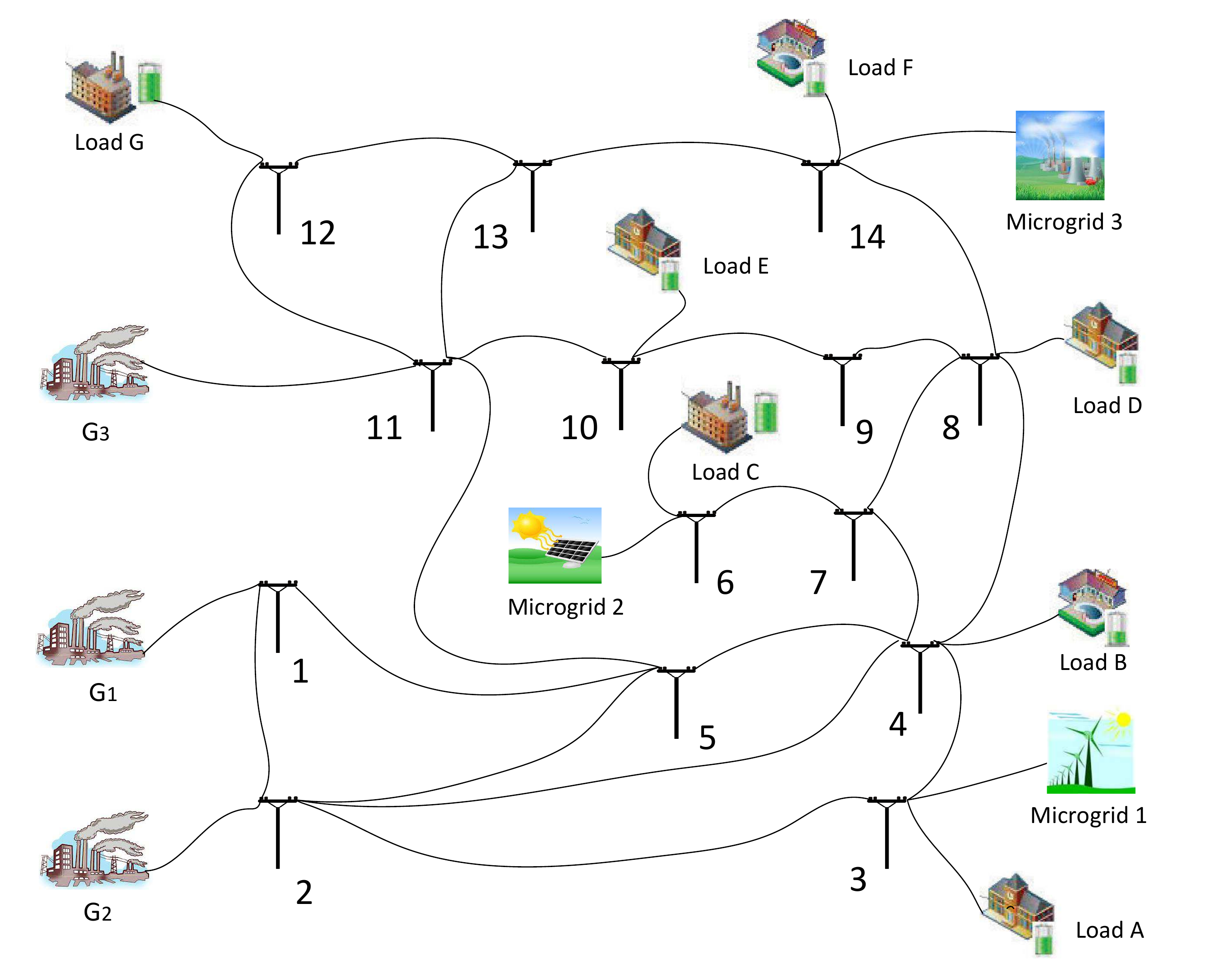}
\caption{IEEE 14-Bus power system model. Buses 3, 6 and 14 are connected  to microgrids that generate renewable energies and are three players in the power generation game.}\label{14bus}
\end{figure}

Based on the grid shown in Fig. \ref{14bus}, we obtain $N_d=3$, and the elements in $\textbf{S}$ corresponding to microgrids constitute a submatrix $$\textbf{S}' = \begin{bmatrix} 0.1212&0.0371&0.0349\\0.0371& 0.3850 & 0.1471\\
0.0349& 0.1471 & 0.3909 \end{bmatrix}.$$ 
Then, based on $\textbf{S}'$, we obtain
$\max_{i,j\neq i\in \mathcal{N}_d} \frac{s_{ij}}{s_{ii}}=0.382$. Therefore, the sufficient condition $\max_{i,j\neq i\in \mathcal{N}_d} \frac{s_{ij}}{s_{ii}}(N_d-1)<1$ is satisfied. The results of the iterative update algorithm are shown in Fig. \ref{iterativeresults}. The algorithm converges to the Nash equilibrium $P_3^{g*}=55.1\rm{MW}$, $P_6^{g*}=34.7\rm{MW}$ and $P_{14}^{g*}=27.9\rm{MW}$ after 7 iterations.
\begin{figure}[t]
  \centering
  \subfigure[]{
    \includegraphics[width=1.65in]{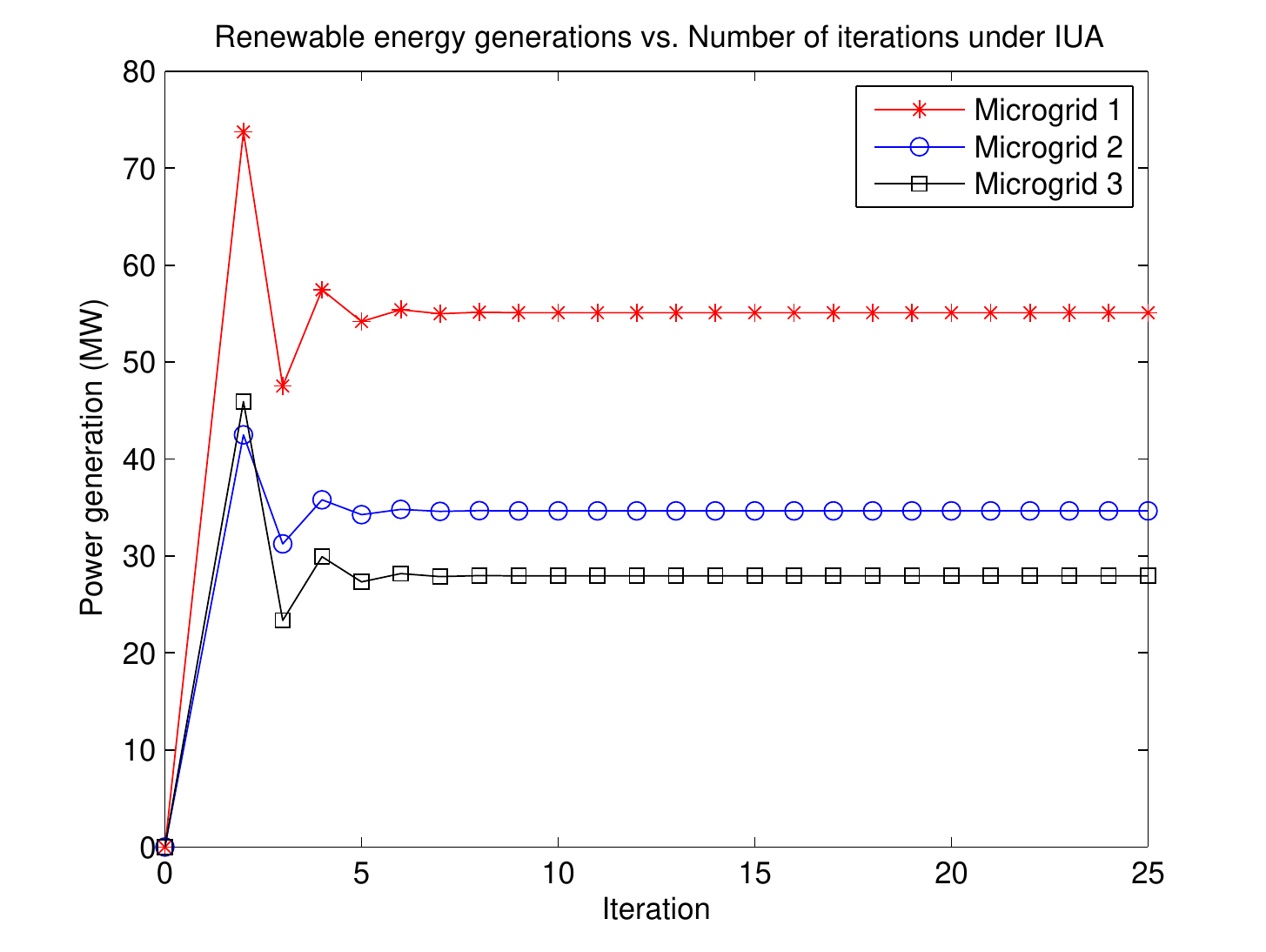}}
	 \subfigure[]{
    \includegraphics[width=1.65in]{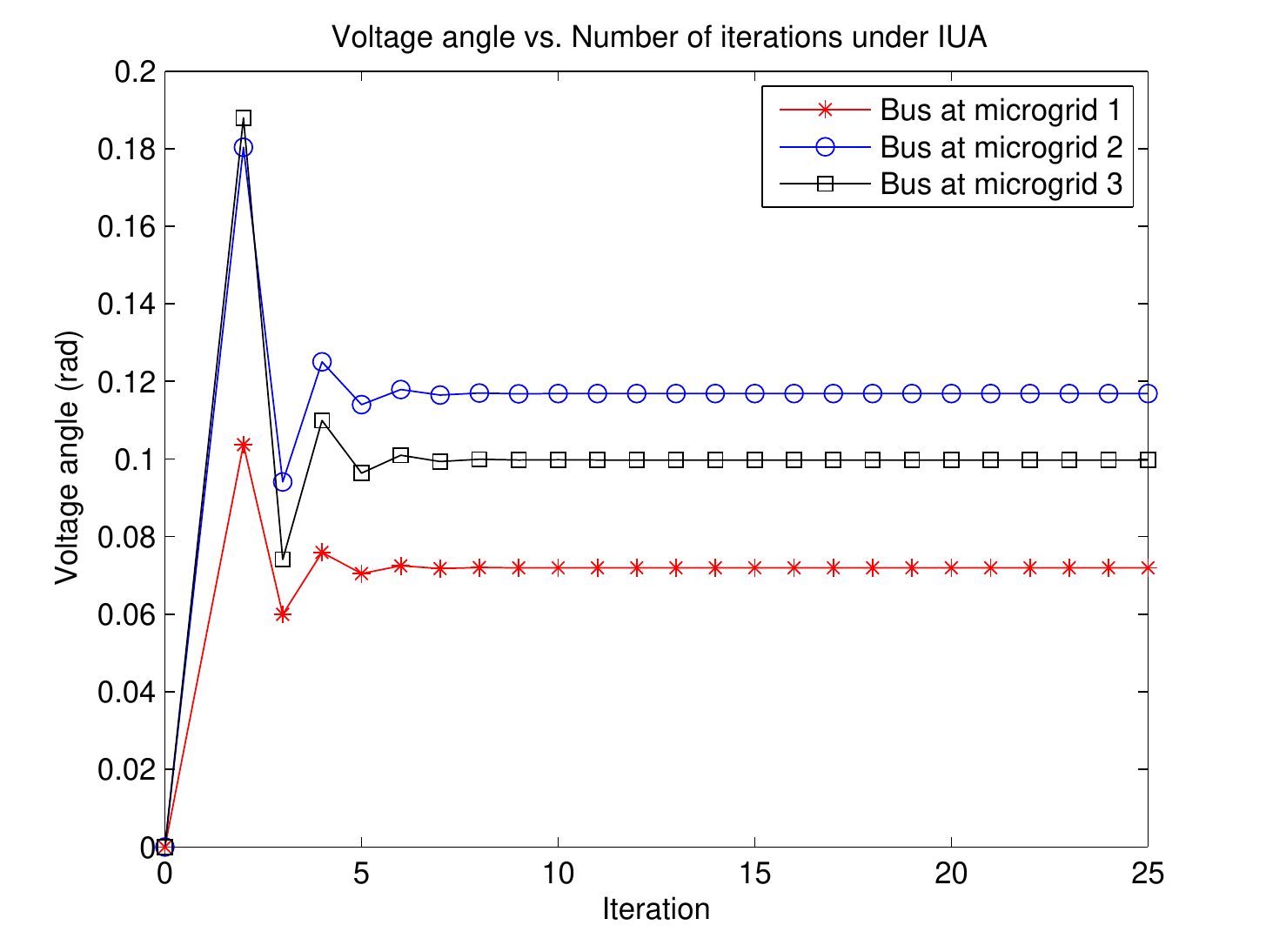}}
  \caption[]{(a) and (b) shows the results of the renewable energy generations and the bus voltage angles of the three microgrids by using IUA, respectively.}
  \label{iterativeresults}
\end{figure}

\begin{figure}[t]
  \centering
  \subfigure[]{
    \includegraphics[width=1.65in]{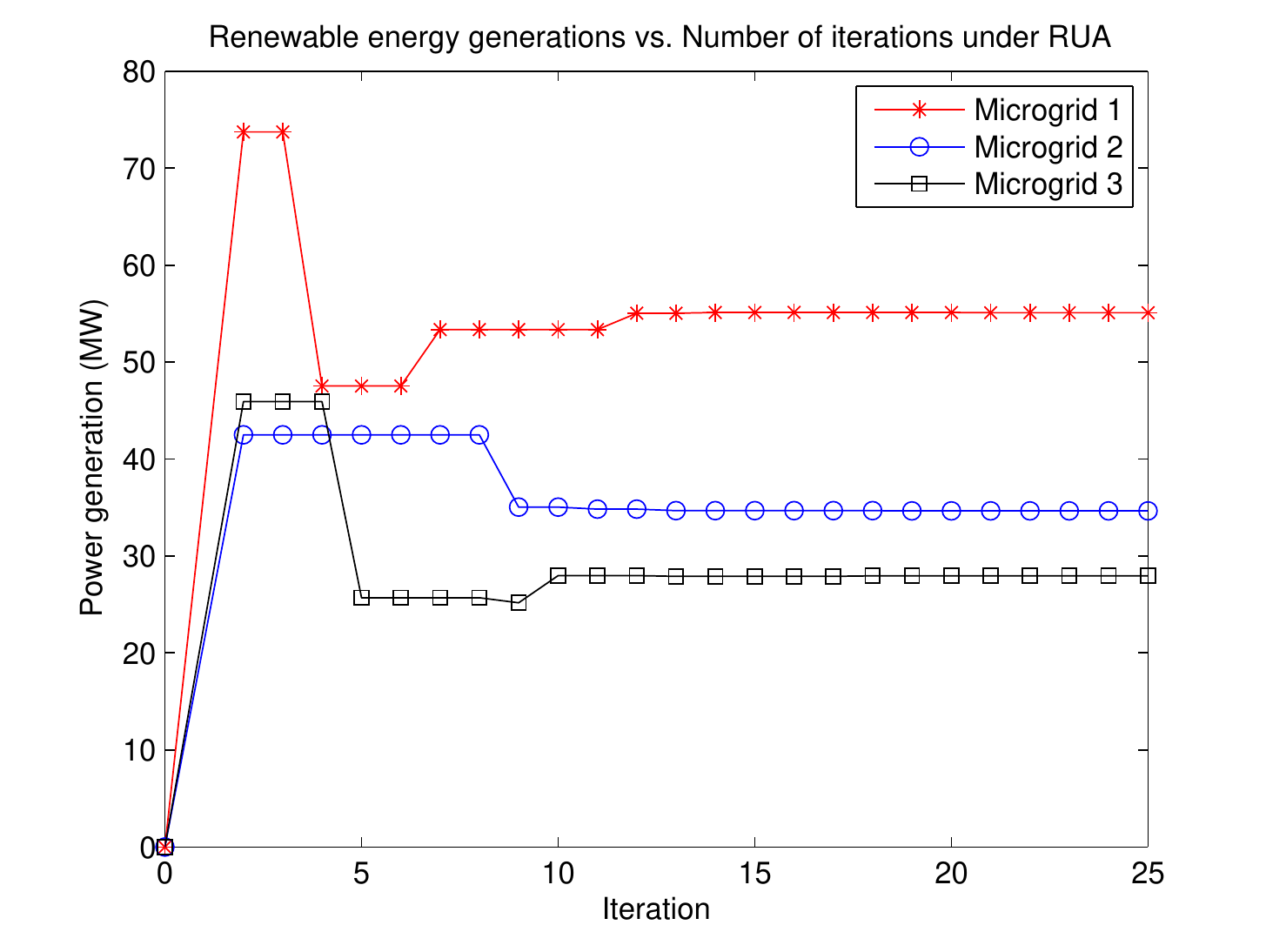}}
	 \subfigure[]{
    \includegraphics[width=1.65in]{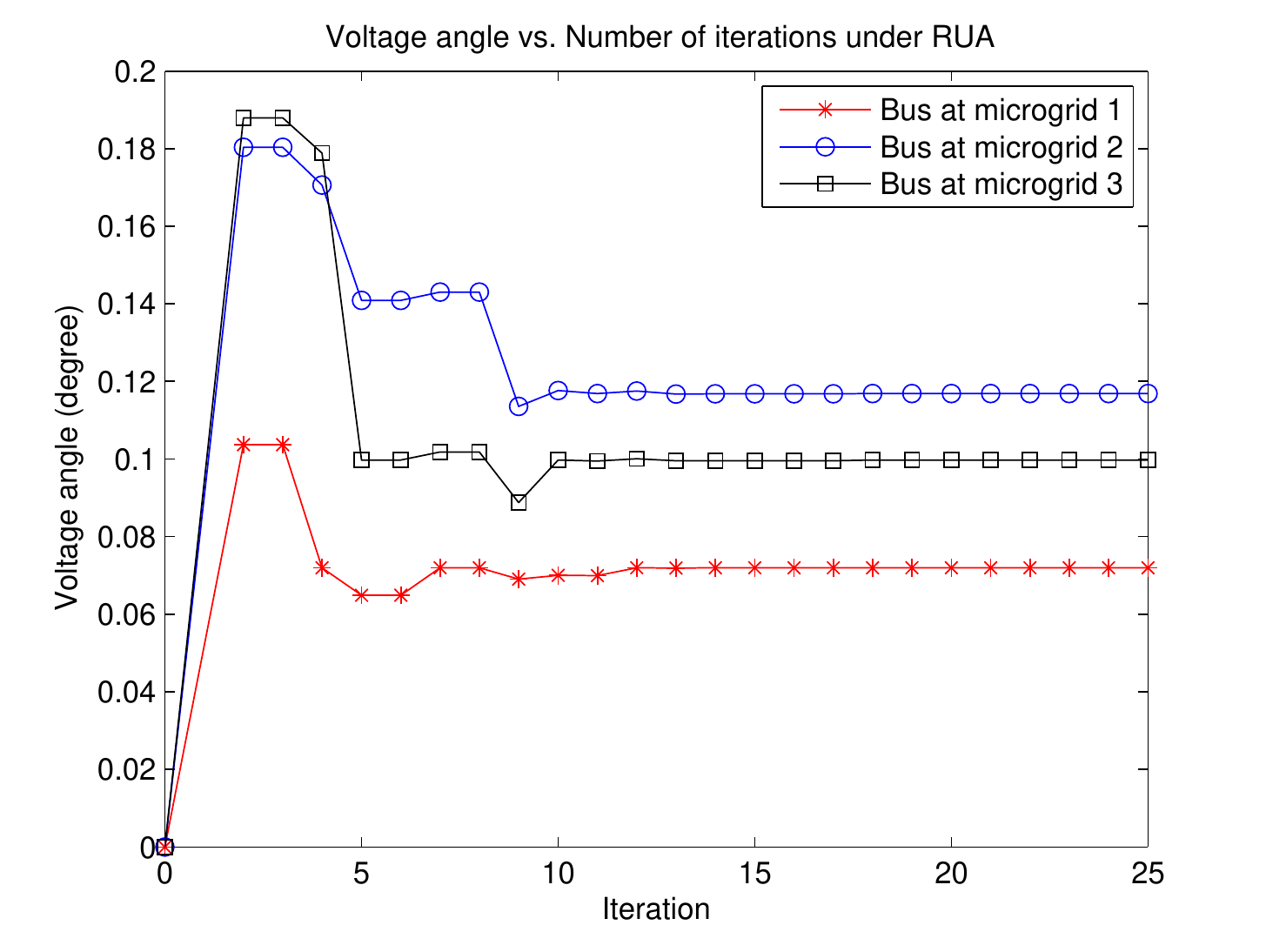}}
  \caption[]{(a) and (b) shows the results of the renewable energy generations and the bus voltage angles of the three microgrids by using RUA, respectively.}
  \label{randomresults}
\end{figure}

To test the RUA, we set $\bar \tau=\underline{\tau}=0.6$ for simplicity, and it also satisfies the sufficient condition in Theorem \ref{thm3}. The results of RUA are shown in Fig. \ref{randomresults}. We can see that the algorithm converges to the Nash equilibrium after 14 iterations. Moreover, the values of $P_3^{g*}$, $P_6^{g*}$ and $P_{14}^{g*}$ at equilibrium are the same as those obtained by using the IUA which validates the effectiveness of RUA.

\subsection{Resiliency of the Distributed Algorithm}
For each update, each player needs to measure the voltage angle at his bus via PMU. In the following case studies, we set $\tau_1=0.65$, $\tau_2=0.7$ and $\tau_3=0.8$ for three players, respectively. Thus, $\bar \tau=0.8$ and $\underline{\tau}=0.65$. We can verify that the sufficient condition in Theorem \ref{thm3} is still satisfied. For a better illustration of the resilience property, we add the failure to the smart grid when the system is at an equilibrium state.

\subsubsection{Generator Breakdown} We consider the scenario that generator $G_3$ connecting with bus 12 is out of service at time step 19. The results are shown in Fig. \ref{gene_down}. The first equilibrium is the same as that obtained by using IUA and RUA. After the fault, the PDA re-converges to a new equilibrium point $P_3^{g*}=60.3\rm{MW}$, $P_6^{g*}=47.8\rm{MW}$ and $P_{14}^{g*}=46.7\rm{MW}$ in another 10 time steps which reveals the resiliency of the algorithm.

\subsubsection{Microgrid Turn-off} When microgrid 3 at bus 14 is turned off at time step 19, the results are shown in Fig. \ref{micro_down}. The distributed algorithm can reach a new equilibrium $P_3^{g*}=62.6\rm{MW}$, $P_6^{g*}=36.5\rm{MW}$ and $P_{14}^{g*}=0\rm{MW}$ only in 4 steps after the shutdown of microgrid 3. The faster convergence rate in this case than that of generator breakdown is due to the number of players is reduced from 3 to 2 which validates the analysis in Section \ref{rateanalysis}.

\subsubsection{Open-circuit Fault of Transmission Line} We consider the scenario that the transmission line connecting buses 8 and 14 in the grid is of open-circuit fault at time step 19. The results for this case are shown in Fig. \ref{line_down}. We can see that PDA is still able to re-converge to an equilibrium $P_3^{g*}=57.9\rm{MW}$, $P_6^{g*}=36.4\rm{MW}$ and $P_7^{g*}=16.4\rm{MW}$ in 7 other time steps after the fault occurs, and the new equilibrium leads to power re-dispatch in the smart grid.

\begin{figure*}[t]
  \centering
  
  \subfigure[]{%
    \includegraphics[width=0.32\textwidth]{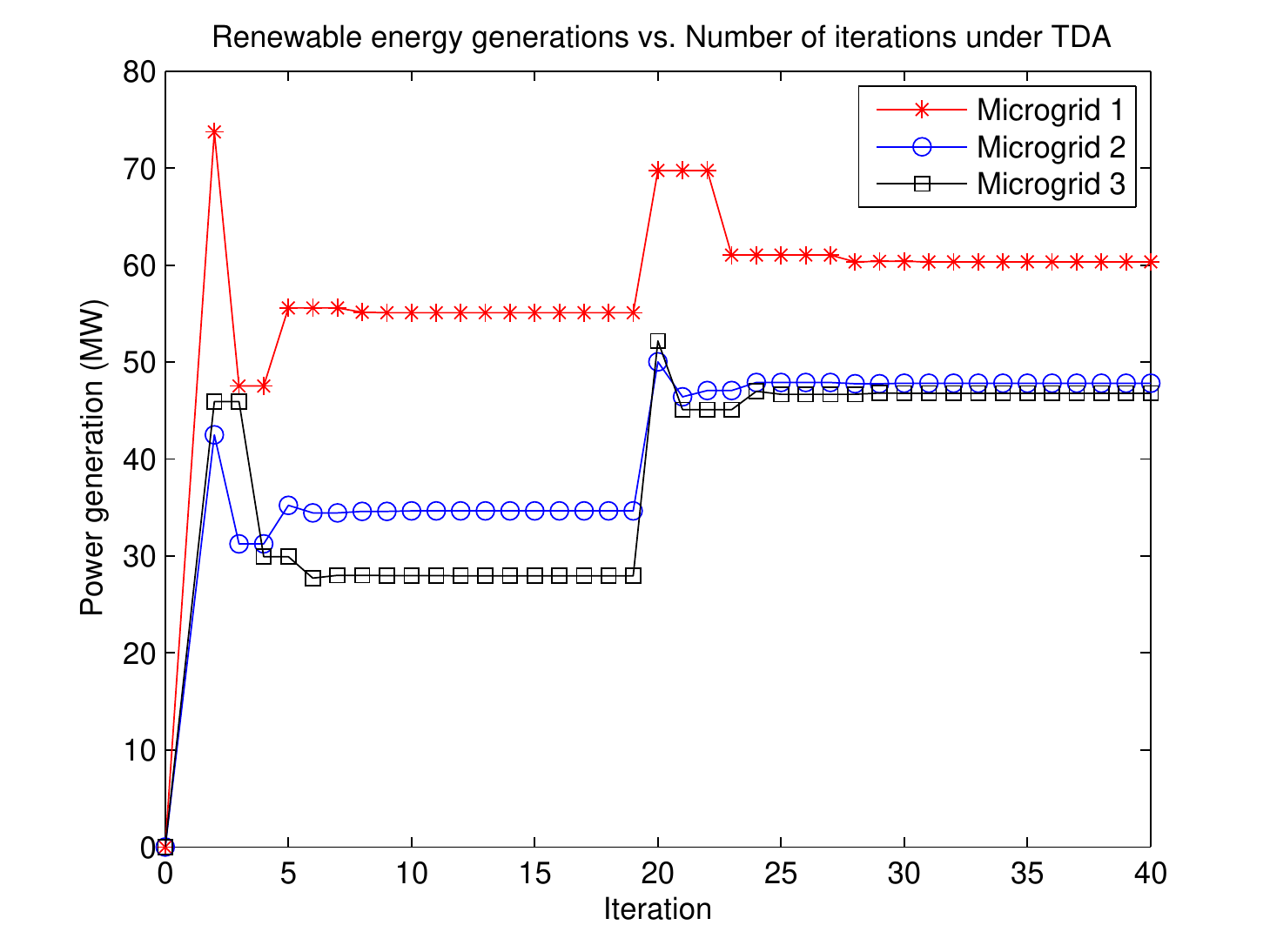}%
    \label{gene_down}%
  }%
  \hfill
  \subfigure[]{%
    \includegraphics[width=0.32\textwidth]{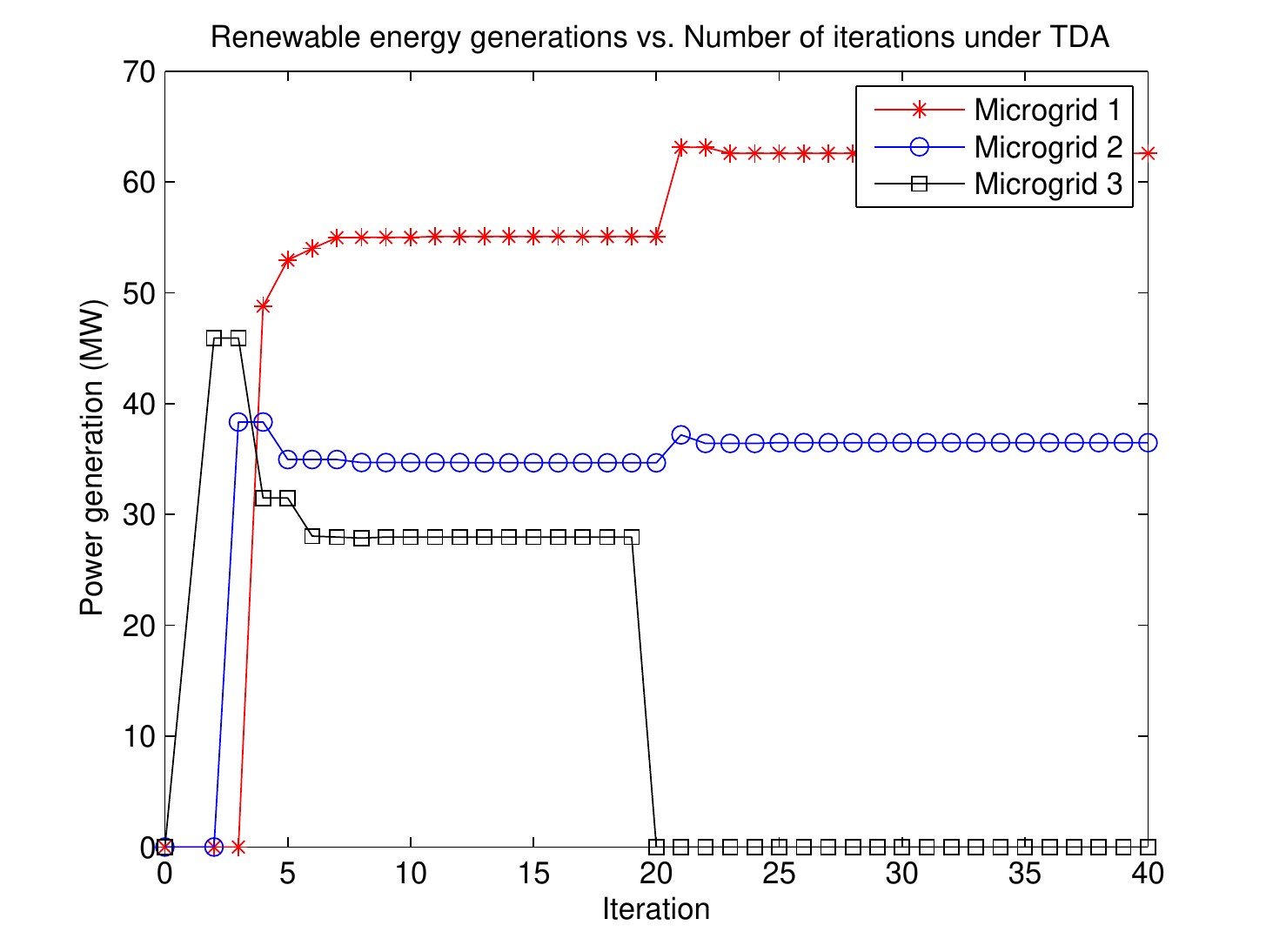}%
    \label{micro_down}%
  }%
  \hfill
  \subfigure[]{%
    \includegraphics[width=0.32\textwidth]{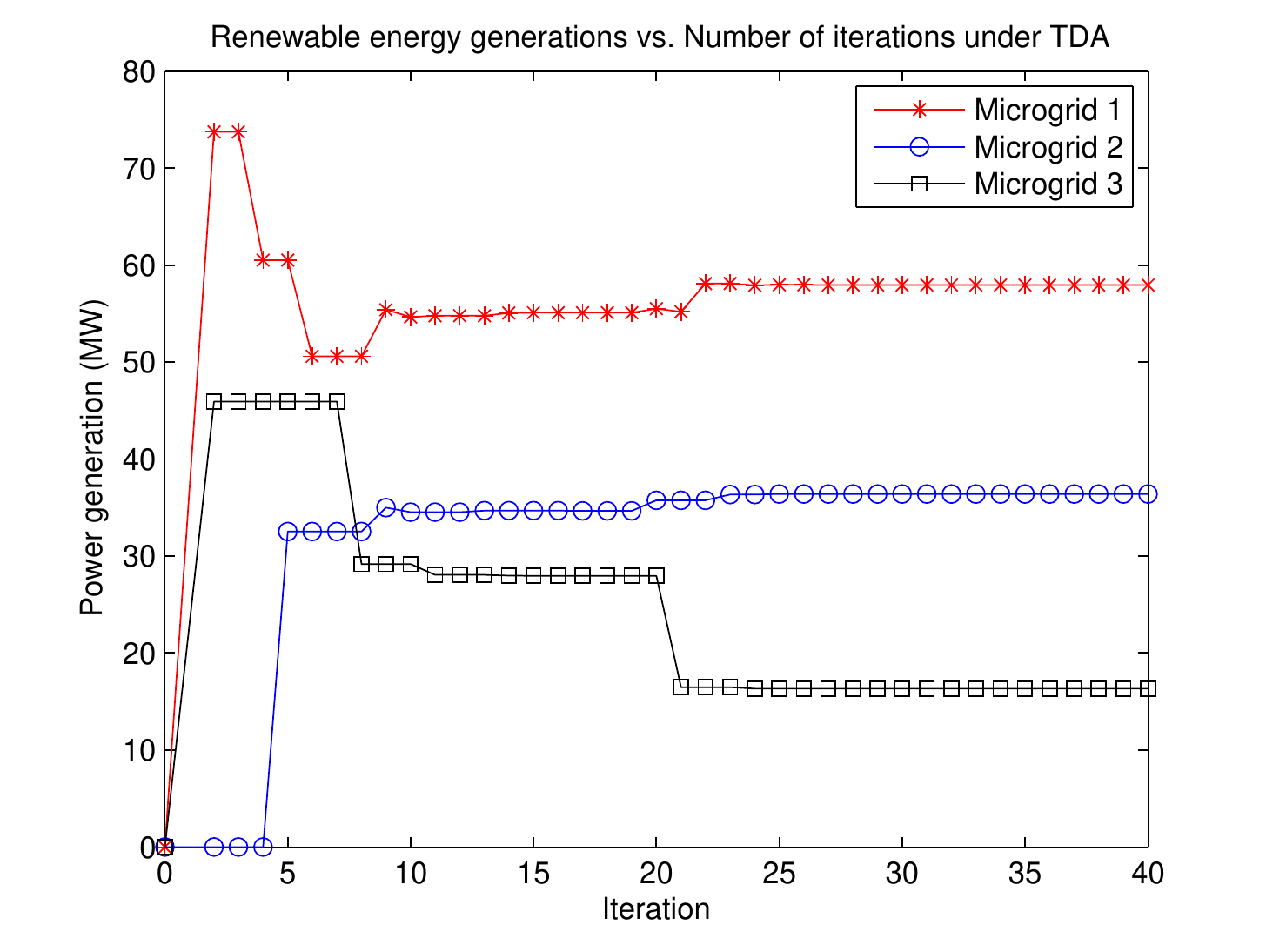}%
    \label{line_down}%
  }%
  \caption{All three faults occur at time step 19. (a), (b) and (c) show the renewable energy generations of three microgrids related to the generator $G_3$ turn-off, microgrid 3 shut down, and open-circuit of the transmission line 8-14 in the smart grid, respectively.}
  \label{bigfig2}
\end{figure*}

\section{Conclusion}\label{conclusion}
In this paper, we have used a game-theoretic framework to capture the interactions between microgrids in a power system. In addition, we have proposed a resilient and fully distributed algorithm for microgrids to update their strategy on the amount of renewable energy generation. The knowledge that microgrid requires is only the voltage angle at his bus which can be measured by PMU. The effectiveness and resiliency of the algorithm have been validated via case studies based on the IEEE 14-bus system. One future work would be considering the generators in the smart grid as leaders, and designing efficient, resilient and distributed algorithms for microgrids. The second extension would be incorporating the storage dynamics into the established framework, and obtaining dynamic optimal power scheduling policies for players. Another future work would be investigating the security of PMU devices, and making the algorithm more resilient when PMU measured data contains errors.

\appendices
\section{Proof of Lemma \ref{uniqueNE}}\label{appLemma}
\begin{proof}
Remind that $-\textbf{B}$ is a real symmetric reduced Laplacian matrix, and $-B_{ij}\leq 0$ for $\forall i\neq j\in\mathcal{N}$. In addition, $\sum_{j\in\mathcal{N}} |B_{ij}|\leq |B_{ii}|$, $\forall i\in\mathcal{N}$. Thus, $-\textbf{B}$ is diagonally dominant. Since $-\textbf{B}$ is invertible, and $-B_{ii}>0$, $\forall i\in\mathcal{N}$, $-\textbf{B}$ is positive definite \cite{berman2003completely}. Together with non-positive off-diagonal elements of $-\textbf{B}$, we know that $-\textbf{B}$ is an $M$-matrix.

Let $c>0$ be the largest diagonal entry of $-\textbf{B}$, then, $-\textbf{B}$ can be expressed as $-\textbf{B}= c\textbf{I}-\textbf{A}$, where $\textbf{I}$ is an identity matrix; and $\textbf{A}$ is a non-negative symmetric matrix. By using the Perron-Frobenius theorem \cite{horn2012matrix}, $\rho(\textbf{A})$ is a positive eigenvalue of $\textbf{A}$, where $\rho(\cdot)$ denotes the spectral radius of a matrix. In addition, for any eigenvalue $\lambda$ of $\textbf{A}$, $c-\lambda$ is an eigenvalue of $-\textbf{B}$. Therefore, $c-\rho(\textbf{A})$ is an eigenvalue of $-\textbf{B}$. Since $-\textbf{B}$ is positive definite, we obtain
\begin{equation}\label{radius}
c-\rho(\textbf{A})>0 \Longleftrightarrow c/\rho(\textbf{A})>1.
\end{equation}

Denote $\tilde{\textbf{A}} = \frac{1}{c}\textbf{A}$, and $-\tilde{\textbf{B}} = -\frac{1}{c}\textbf{B}$. Then, $-\tilde{\textbf{B}} = \textbf{I}-\tilde{\textbf{A}}$. Note that $\rho(\tilde{\textbf{A}})=\frac{1}{c}\rho(\textbf{A})<1$ by \eqref{radius}. The summation of infinite series
\begin{align*}
-\tilde{\textbf{B}}\sum_{n=0}^\infty (\tilde{\textbf{A}})^n = (\textbf{I}-\tilde{\textbf{A}})\sum_{n=0}^\infty (\tilde{\textbf{A}})^n &= \sum_{n=0}^\infty(\tilde{\textbf{A}})^n - \sum_{n=1}^\infty(\tilde{\textbf{A}})^n\\
& = (\tilde{\textbf{A}})^0 = \textbf{I}.
\end{align*}
Therefore, we obtain $-(\tilde{\textbf{B}})^{-1} = \sum_{n=0}^\infty (\tilde{\textbf{A}})^n$, and
$$-\textbf{B}^{-1} = -(c\tilde{\textbf{B}})^{-1}=-\frac{1}{c}(\tilde{\textbf{B}})^{-1} = \frac{1}{c}\sum_{n=0}^\infty (\tilde{\textbf{A}})^n.$$
Note that $-\textbf{B}^{-1}=[s_{ij}]_{i,j\in \mathcal{N}}$ is a summation of a series of non-negative symmetric matrices, and thus $-\textbf{B}^{-1}$ is symmetric, and ${s_{ij}\geq 0},\ \forall i,j\in\mathcal{N}$. Furthermore, since $(\tilde{\textbf{A}})^0 = \textbf{I}$, and $\sum_{n=1}^\infty(\tilde{\textbf{A}})^n$ is non-negative, then, ${s_{ii}> 0},\ \forall\ i\in\mathcal{N}$.
\end{proof}

\section{Proof of Theorem \ref{uniqueNE}}\label{app1}
\begin{proof}
When the equilibrium solution $\textbf{P}_d^{*}$ is an inner point, to show its uniqueness, one way is to show that matrix \textbf{H} is invertible. Since $\textbf{S}$ is of full rank, and base on the Sylvester's criterion, the upper left $N_d$-dimensional square matrix $\textbf{S}_1$ in $\textbf{S}$ is also invertible. Note that determinant $|\textbf{S}_1|\neq 0$ and it satisfies $|\textbf{S}_1|=|\textbf{H}|\cdot \prod_{i\in\mathcal{N}_d} s_{ii}$. Because $s_{ii}>0,\ \forall i\in\mathcal{N}_d$, $|\textbf{H}|\neq 0$ and thus \textbf{H} is invertible. Therefore, game $\widetilde G$ admits a unique Nash equilibrium in this case.

Then, we consider the case that some players achieve the boundary at the equilibrium and show that the boundary solution is unique. For convenience, we reorganize the players' indexing as follows: players $\{1,2,...,M_1\}$, $\{M_1+1,...,M_2\}$ and $\{M_2+1,...,N_d\}$ are with inner, zero and maximum power generation at the equilibrium, respectively. When $M_2=N_d$, no player generates the maximum power at the equilibrium, while $M_1=M_2$ indicates that all players generate power at the equilibrium. In addition, we delete the rows and columns corresponding to the players with the boundary power generation in \eqref{matrixform} which yields
\begin{equation*}
\begin{bmatrix}
1&\frac{s_{12}}{s_{11}}&\frac{s_{13}}{s_{11}}& \dotsm &\frac{s_{1M_1}}{s_{11}}\\[6pt]
\frac{s_{21}}{s_{22}}&1&\frac{s_{23}}{s_{22}}& \dotsm &\frac{s_{2M_1}}{s_{22}}\\[6pt]
\vdots & \vdots & \vdots & \ddots & \vdots\\[6pt]
\frac{s_{M_1 1}}{s_{M_1 M_1}}& \frac{s_{M_1 2}}{s_{M_1 M_1}} & \dotsm & \frac{s_{M_1 (M_1-1)}}{s_{M_1 M_1}} & 1
\end{bmatrix}
\begin{bmatrix}
P_1^*\\[6pt]
P_2^*\\[6pt]
\vdots\\[6pt]
P_{M_1}^*
\end{bmatrix}= \textbf{q}_{M_1}^{\dagger}
\end{equation*}
\begin{flalign}\label{matrixform2}
&\Leftrightarrow \quad \textbf{H}_{M_1}\textbf{P}_{d,M_1}^{*}=\textbf{q}_{M_1}^{\dagger},&
\end{flalign}
where $\textbf{q}_{M_1}^{\dagger}:=[q_1^{\dagger},q_2^{\dagger},...,q_{M_1}^{\dagger}]'$,
$
q_i^{\dagger}:=q_i-\frac{1}{s_{ii}}\big(s_{i(M_1+1)}P_{M_1+1}^{\max}
+s_{i(M_1+2)}P_{M_1+2}^{\max}+...+s_{iM_2}P_{M_2}^{\max}\big)
$
for $i\in\{1,2,...,M_1\}$. It directly follows that $\textbf{H}_{M_1}$ is invertible, and thus the solution $\textbf{P}_{d,M_1}^{*}$ exists and is unique. The reason that the boundary solution is an equilibrium is as follows. The payoff function \eqref{utility2} for player $i$ is quadratic and convex. For the cases $\gamma_i \leq \bar{g}_{-i}-s_{ii}P_i^l$ and $\gamma_i \geq \bar{g}_{-i}+s_{ii}P_i^{\max}$, player $i$'s optimal response with respect to $P_{-i}^g$ is achieved at $P_i^{g*}=0$ and $P_i^{g*}=P_i^{\max}$, respectively. 

Next, we argue that the boundary solution is unique. One possible case is that player $i$, $i\in\{1,2,...,M_1\}$, achieves the boundary power at the equilibrium. This case is ruled out when we construct \eqref{matrixform2}. Another possible case is that player $j$, $j\in\{M_1+1,M_1+2,...,N_d\}$, has the inner power generation. This cannot be a Nash equilibrium either, since player $j$ can achieve a better payoff at the boundary. Other possible cases can be a combination of these two cases and can be easily eliminated through similar analysis. Therefore, the boundary Nash equilibrium solution is unique.
\end{proof}

\section{Proof of Theorem \ref{stability}}\label{app2}
\begin{proof}
First, define $\Delta P_i^{(n)}:= P_i^{(n)}-P_i^*,\ i\in \mathcal{N}_d$. When $P_i^{\max}>P_i^*>-P_i^l,\ \forall i\in \mathcal{N}_d$, and given all players' generated renewable energy  except $i$th player's at time step $n$, we have the following:

(1) $\bar{g}_{-i}^{(n)}-s_{ii}P_i^l< \gamma_i  < \bar{g}_{-i}^{(n)}+s_{ii}P_i^{\max}$: 
\begin{align*}
\Delta P_i^{(n+1)}&= P_i^{(n+1)}-P_i^*\\
&=\frac{1}{s_{ii}}\big(\gamma_i - \bar{g}_{-i}^{(n)}\big)-\frac{1}{s_{ii}}\big(\gamma_i - \bar{g}_{-i}^{*}\big)\\
&=\frac{1}{s_{ii}} \sum\limits_{j\neq i \in \mathcal{N}_d}s_{ij}\big(P_j^*-P_j^{(n)} \big)\\
&=-\frac{1}{s_{ii}} \sum\limits_{j\neq i \in \mathcal{N}_d}s_{ij}\Delta P_j^{(n)}.
\end{align*}

(2) $\gamma_i\leq \bar{g}_{-i}^{(n)}-s_{ii}P_i^l$:
\begin{align*}
\Delta P_i^{(n+1)} &= -P_i^l-P_i^*\qquad (\mathrm{negative})\\
&>\frac{1}{s_{ii}}\big(\gamma_i - \bar{g}_{-i}^{(n)}\big)-\frac{1}{s_{ii}}\big(\gamma_i - \bar{g}_{-i}^{*}\big)\\
&=-\frac{1}{s_{ii}} \sum\limits_{j\neq i \in \mathcal{N}_d}s_{ij}\Delta P_j^{(n)}.
\end{align*}

(3) $\gamma_i\geq \bar{g}_{-i}^{(n)}+s_{ii}P_i^{\max}$:
\begin{align*}
\Delta P_i^{(n+1)}&= P_i^{\max}- P_i^*\qquad (\mathrm{positive})\\
 &<\frac{1}{s_{ii}}\big(\gamma_i - \bar{g}_{-i}^{(n)}\big)-\frac{1}{s_{ii}}\big(\gamma_i - \bar{g}_{-i}^{*}\big)\\
&=-\frac{1}{s_{ii}} \sum\limits_{j\neq i \in \mathcal{N}_d}s_{ij}\Delta P_j^{(n)}.
\end{align*}

To sum up,
\begin{align*}
|\Delta P_i^{(n+1)}|\begin{cases}
\begin{array}{l}
<\sum\limits_{j\neq i \in \mathcal{N}_d}\frac{s_{ij}}{s_{ii}}|\Delta P_j^{(n)}|, \mathrm{if}\ \gamma_i\leq \bar{g}_{-i}^{(n)}-s_{ii}P_i^l,\\
\qquad\qquad\qquad\quad\ \ \mathrm{or}\ \gamma_i\geq \bar{g}_{-i}^{(n)}+s_{ii}P_i^{\max},\\
= \sum\limits_{j\neq i \in \mathcal{N}_d}\frac{s_{ij}}{s_{ii}}|\Delta P_j^{(n)}|,\qquad\qquad \mathrm{otherwise}.
\end{array}
\end{cases}
\end{align*}

When the generated renewable energy of an arbitrary player $i$ at the Nash equilibrium point is 0, i.e., $P_i^*=-P_i^l$, then through similar analysis, we have
\begin{align*}
|\Delta P_i^{(n+1)}|
\begin{cases}
\begin{array}{ll}
= 0, \qquad\qquad\qquad\quad \mathrm{if}\ \gamma_i \leq \bar{g}_{-i}-s_{ii}P_i^l,\\
<\sum\limits_{j\neq i \in \mathcal{N}_d}\frac{s_{ij}}{s_{ii}}|\Delta P_j^{(n)}|, \qquad\qquad \mathrm{otherwise}.
\end{array}
\end{cases}
\end{align*}

In addition, when the generated power of player $i$ at the Nash equilibrium point is $P_{i,\max}^g$, i.e., $P_i^*=P_i^{\max}$, then,
\begin{align*}
|\Delta P_i^{(n+1)}|\begin{cases}
\begin{array}{ll}
= 0,\qquad\qquad\qquad \mathrm{if}\ \gamma_i \geq \bar{g}_{-i}+s_{ii}P_i^{\max},\\
<\sum\limits_{j\neq i \in \mathcal{N}_d}\frac{s_{ij}}{s_{ii}}|\Delta P_j^{(n)}|,\qquad\qquad \mathrm{otherwise}.
\end{array}
\end{cases}
\end{align*}
Therefore, 
\begin{equation}\label{contraction}
|\Delta P_i^{(n+1)}|\leq \frac{1}{s_{ii}} \sum\limits_{j\neq i \in \mathcal{N}_d}s_{ij}|\Delta P_j^{(n)}|
\end{equation}
 holds for all three cases at any time step $n$.

Next, let $\lVert \Delta P \rVert_\infty$ denote the infinity-norm of the vector $(\Delta P_1,\Delta P_2,...,\Delta P_{N_d})^T$, i.e., $\lVert \Delta P \rVert_\infty= \max_{i \in\mathcal{N}_d} |\Delta P_i|$. Then, we have
\begin{align*}
\lVert \Delta P^{(n+1)} \rVert_\infty &\leq \max_{i\in\mathcal{N}_d} \big\{  \frac{1}{s_{ii}} \sum_{j\neq i\in \mathcal{N}_d}s_{ij}|\Delta P_j^{(n)}| \big\}\\
&\leq \max_{i,j\neq i \in\mathcal{N}_d} \frac{s_{ij}}{s_{ii}}(N_d-1) \lVert \Delta P^{(n)} \rVert_\infty.
\end{align*}
Therefore,
$\max_{i,j\neq i\in \mathcal{N}_d} \frac{s_{ij}}{s_{ii}}(N_d-1)<1 $ is a sufficient condition under which \eqref{contraction} is a \textit{contraction mapping} that ensures the global stability and convergence of IUA.
\end{proof}

\section{Proof of Theorem \ref{thm3}}\label{app3}
\begin{proof}
Based on \eqref{contraction}, we have the following
\begin{align*}
&\mathbb{E}(|\Delta P_i^{(n+1)}|)\\
&= \mathbb{E}\Big( |\Delta P_i^{(n+1)}|\Big|\ P_i\ \mathrm{updates\ at\ time}\ n \Big)\cdot\tau_i\\
&+\mathbb{E}\Big( |\Delta P_i^{(n)}|\Big|\ P_i\ \mathrm{does\ not\ update\ at\ time}\ n \Big)\cdot(1-\tau_i)\\
&\leq \frac{\tau_i}{s_{ii}} \sum\limits_{j\neq i \in \mathcal{N}_d}s_{ij}\cdot\mathbb{E}(|\Delta P_i^{(n)}|)+(1-\tau_i)\cdot\mathbb{E}(|\Delta P_i^{(n)}|).
\end{align*}
Define $\lVert \Delta P \rVert_\infty:= \max_{i \in\mathcal{N}_d} \mathbb{E}(|\Delta P_i|)$. Then,
\begin{align*}
\max_{i \in\mathcal{N}_d}\ \mathbb{E}(|\Delta P_i^{(n+1)}|)\leq  \max_{i,j\neq i\in \mathcal{N}_d} \tau_i \frac{s_{ij}}{s_{ii}}(N_d-1)\lVert \Delta P^{(n)} \rVert_\infty&\\
\qquad+\max_i\ (1-\tau_i)\lVert \Delta P^{(n)} \rVert_\infty &,
\end{align*}
which is equivalent to
\begin{align}
\lVert \Delta P&^{(n+1)} \rVert_\infty\leq \notag\\
& \Big( \bar\tau\cdot \max_{i,j\neq i\in \mathcal{N}_d} \frac{s_{ij}}{s_{ii}}(N_d-1)+ (1-\underline\tau) \Big)\lVert \Delta P^{(n)} \rVert_\infty.\label{contractionrandom}
\end{align}
Therefore, $\bar\tau\cdot \max_{i,j\neq i\in \mathcal{N}_d} \frac{s_{ij}}{s_{ii}}(N_d-1) + (1-\underline\tau)<1$ implies
\begin{align}
\bar{\tau}\cdot \max_{i,j\neq i\in \mathcal{N}_d} \frac{s_{ij}}{s_{ii}}(N_d-1)<\underline{\tau}\label{convergecondition}
\end{align} which is a sufficient condition that leads the right hand side in \eqref{contractionrandom} to a contraction mapping, and thus ensures the stability and convergence of the random update algorithm in the infinity norm. Next, we show a stronger convergence, \textit{almost sure} convergence, of the algorithm under \eqref{convergecondition}.

From \eqref{contractionrandom}, we have
\begin{align}
\lVert \Delta P^{(n)} \rVert_\infty \leq \xi\lVert\Delta P^{(n-1)}\rVert_\infty\leq...\leq \xi^{n}\lVert \Delta P^{(0)} \rVert_\infty,\label{inequality}
\end{align}
where $0<\xi<1$. Moreover, by using the Markov inequality, we obtain
\begin{align}
\sum_{n=1}^{\infty}\Pr\big( |\Delta P_i^{(n)}|> \epsilon \big) &\leq \sum_{n=1}^{\infty} \frac{\mathbb{E}(|\Delta P_i^{(n)}|)}{\epsilon}\notag\\
&\leq \frac{1}{\epsilon}\sum_{n=1}^{\infty} \lVert \Delta P^{(n)} \rVert_\infty,\label{markovinequality}
\end{align}
where $\epsilon>0$ and $\Pr(\cdot)$ is the probability measure. Inserting \eqref{inequality} into \eqref{markovinequality} yields
\begin{align*}
\sum_{n=1}^{\infty}\Pr\big( |\Delta P_i^{(n)}|> \epsilon \big) &\leq \frac{1}{\epsilon} \sum_{n=1}^{\infty}  \xi^{n}\lVert \Delta P^{(0)} \rVert_\infty\\
&=\frac{1-\xi^{\infty}}{\epsilon(1-\xi)} \lVert \Delta P^{(0)} \rVert_\infty\\
&=\frac{1}{\epsilon (1-\xi)}  \lVert \Delta P^{(0)} \rVert_\infty.
\end{align*}
Therefore, the increasing sequence $\sum_{n=1}^{N}\Pr( |\Delta P_i^{(n)}|> \epsilon)$ is upper bounded by $\frac{1}{\epsilon (1-\xi)}  \lVert \Delta P^{(0)} \rVert_\infty$ with the augment of $N$, which indicates the convergence of the sequence $\sum_{n=1}^{N}\Pr( |\Delta P_i^{(n)}|> \epsilon)$  for $\forall \epsilon>0$. By using the Borel-Cantelli lemma \cite{billingsley2008probability}, we obtain
$$\Pr\Big( \limsup_{n\rightarrow\infty}\{|\Delta P_i^{(n)}|> \epsilon \} \Big)=0,\ \forall i\in \mathcal{N}_d,$$
which implies that the proposed random update algorithm converges to the unique Nash equilibrium almost surely when \eqref{convergecondition} is satisfied.
\end{proof}

\bibliographystyle{IEEEtran}
\bibliography{IEEEabrv,references}

\end{document}